\documentclass[sigconf]{acmart}

 \AtBeginDocument{%
 }

\copyrightyear{2023}
\acmYear{2023}
\setcopyright{rightsretained}
\acmConference[KDD '23]{Proceedings of the 29th ACM SIGKDD Conference on Knowledge Discovery and Data Mining}{August 6--10, 2023}{Long Beach, CA, USA}
\acmBooktitle{Proceedings of the 29th ACM SIGKDD Conference on Knowledge Discovery and Data Mining (KDD '23), August 6--10, 2023, Long Beach, CA, USA}
\acmDOI{10.1145/3580305.3599434}
\acmISBN{979-8-4007-0103-0/23/08}

\makeatletter
\gdef\@copyrightpermission{
  \begin{minipage}{0.3\columnwidth}
   \href{https://creativecommons.org/licenses/by/4.0/}{\includegraphics[width=0.90\textwidth]{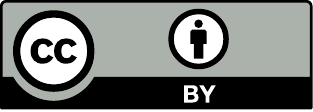}}
  \end{minipage}\hfill
  \begin{minipage}{0.7\columnwidth}
   \href{https://creativecommons.org/licenses/by/4.0/}{This work is licensed under a Creative Commons Attribution International 4.0 License.}
  \end{minipage}
  \vspace{5pt}
}
\makeatother

\usepackage{amsmath}
\usepackage{amsthm}
\usepackage{tikz}
\usepackage{graphicx}
\usepackage{hyperref}
\usepackage{epstopdf}
\usepackage{bbm}
\usepackage{enumerate}
\usepackage{footnote}
\usepackage{url}
\usepackage{algorithm}
\usepackage[noend]{algpseudocode}
\usepackage{xspace}
\usepackage[caption=false,font=small]{subfig}
\usepackage{balance}

\usetikzlibrary{positioning} 
\usetikzlibrary{patterns}
\usetikzlibrary{fit,shapes.geometric}
 
\DeclareMathOperator*{\argmin}{arg\,min}	
\DeclareMathOperator*{\argmax}{arg\,max}
\DeclareMathOperator{\E}{\mathbb{E}}

\newcommand{\BMMH}{\textsf{BMMH}\xspace}
\newcommand{\BMAH}{\textsf{BMAH}\xspace}
\newcommand{\greedy}{\texttt{Greedy}\xspace}
\newcommand{\greedyplus}{\texttt{Greedy+}\xspace}

\newcommand{\bigO}{\ensuremath{\mathcal{O}}\xspace}
\newcommand{\NP}{\ensuremath{\mathbf{NP}}\xspace}

\newtheorem{cor}{\textbf{Corollary}}
\newtheorem{fact}{\textbf{Fact}}
\newtheorem{defi}{\textbf{Definition}}   
\newtheorem{theorem}{\textbf{Theorem}}    
\newtheorem{lemma}{\textbf{Lemma}}     
\newtheorem{proposition}{\textbf{Proposition}}
\newtheorem{prob}{\textbf{Problem}}
\newtheorem{obs}{\textbf{Observation}}

\newcommand{\norm}[1]{\left\lVert#1\right\rVert}
\newcommand{\ceil}[1]{\lceil {#1} \rceil}
\algblockdefx{MRepeat}{EndRepeat}{\textbf{repeat $k$ times}}{}
\algnotext{EndRepeat}

\newcommand{\ptitle}[1]{\smallskip\noindent{\bf #1.}}
\newcommand{\pttitle}[1]{\smallskip\noindent{\it #1.}}

\begin{document}

\title[Minimizing Hitting Time between Disparate Groups with Shortcut Edges]%
{Minimizing Hitting Time between Disparate Groups\\with Shortcut Edges}

\author{Florian Adriaens}
\email{florian.adriaens@helsinki.fi}
\affiliation{%
  \institution{University of Helsinki}
  \city{Helsinki}
  \country{Finland}
}

\author{Honglian Wang}
\email{honglian@kth.se}
\affiliation{%
  \institution{KTH Royal Institute of Technology}
  \city{Stockholm}
  \country{Sweden}
}

\author{Aristides Gionis}
\email{argioni@kth.se}
\affiliation{%
  \institution{KTH Royal Institute of Technology}
  \city{Stockholm}
  \country{Sweden}
}

\renewcommand{\shortauthors}{Florian Adriaens, Honglian Wang, \& Aristides Gionis}

\begin{abstract}
Structural bias or segregation of networks refers to situations where
two or more disparate groups are present in the network, 
so that the groups are highly connected internally, but loosely connected to each other.
Examples include polarized communities in social networks, 
antagonistic content in video-sharing or news-feed platforms, etc.
In many cases it is of interest to increase the connectivity of disparate groups
so as to, e.g., minimize social friction, or expose individuals to diverse viewpoints. 
A commonly-used mechanism for increasing the network connectivity 
is to add \emph{edge shortcuts} between pairs of nodes.
In many applications of interest, 
edge shortcuts typically translate to \emph{recommendations}, 
e.g., what video to watch, or what news article to read next.
The problem of reducing structural bias or segregation via edge shortcuts
has recently been studied in the literature, 
and random walks have been an essential tool for modeling 
navigation and connectivity in the underlying networks.
Existing methods, however, either do not offer approximation guarantees,
or engineer the objective so that it satisfies certain desirable properties 
that simplify the optimization~task.

In this paper we address the problem of adding a given number of shortcut
edges in the network so as to \emph{directly} minimize
the \emph{average hitting time} and the \emph{maximum hitting time}
between two disparate groups. 
The objectives we study are more natural than 
objectives considered earlier in the literature 
(e.g., maximizing hitting-time \emph{reduction})
and the optimization task is significantly more challenging. 
Our algorithm for minimizing average hitting time is a greedy
bi\-criteria that relies on super\-modularity. 
In contrast, maximum hitting time is not supermodular.
Despite, we develop an approximation algorithm for that objective as well,  
by leveraging connections 
with average hitting time and the asymmetric $k$-center~problem.
\end{abstract}



\maketitle

\section{Introduction}
\label{sec:intro}

The last decade has seen a surge in the development of methods for detecting, 
quantifying, and mitigating polarization and controversy in social media.
An example of a polarized network is a conversation or endorsement graph 
associated with a controversial topic on Twitter. 
The nodes represent users, and the edges interactions between users in the form of posts, likes, endorsements, etc., related to a controversial topic. 
It has been observed that the interactions typically occur between like-minded individuals, resulting in the reinforcement of ones own~beliefs~\cite{barbera2020social,chitra2019understanding}.

\emph{Structural bias} arises in many types of networks beyond social networks, 
such as \emph{content} and \emph{information networks}.
Recently, several methods have been developed for reducing \emph{structural bias} \cite{repbub} 
or \emph{segregation} \cite{coupette2023reducing,fabbri2022rewiring} in \emph{content networks}. 
An example of a content network is the network obtained by ``what to watch next'' 
recommendations in a video-sharing platform.
Content in such networks can often be divided into two or more groups, 
each of which is highly connected internally, but loosely connected to each other.
These groups could be the result of differentiating between ``harmful'' (or radicalized) content  
and ``neutral'' (or non-radicalized) content, or simply 
the result of different opinion (pro/contra) on a certain issue.
It might be beneficial for users navigating these platforms to be exposed to diverse content---exposing themselves to multiple viewpoints---in order to become better informed.

A particular line of research focuses on adding new edges 
(denoted as \emph{shortcut} edges) to a network, 
with the goal of reducing some quantifiable measure of polarization, 
controversy, structural bias, or segregation~\cite{coupette2023reducing,demaine2010minimizing,fabbri2022rewiring,repbub,garimella2018quantifying,parotsidis2015selecting}.
Many of these measures are defined on the basis of \emph{random walks} between groups of nodes.
In content graphs, random walks are a natural and simple model for how a user navigates through content.
For example, 
the Random Walk Controversy (RWC) score of Garimella et al.~\cite{garimella2018quantifying} 
has been reported to ``discriminate controversial topics with great accuracy'' 
in Twitter conversation graphs. 
The authors in that work suggest a method that adds a fixed number of shortcut edges 
to a graph in order to reduce the RWC score of the augmented graph~\cite{garimella2018quantifying}.

However, most of these random walk-based augmentation problems 
do not have algorithms with provable approximation guarantees~\cite{garimella2018quantifying}, 
or the optimized measure has been reformulated 
so that it exhibits a desirable~property~\cite{coupette2023reducing,fabbri2022rewiring,repbub}.

For example, the work of Haddadan et al.~\cite{repbub} considers maximizing the gain function $\Delta_g({F}) = g(\emptyset)-g(F)$ in the context of bounded-length random walks in directed graphs.
Here, $F$ is a set of at most $k$ new shortcut edges that will be added to an input graph, and the function $g$ is the \emph{average expected hitting time} for a random walk starting in one group to hit the other group (see Definition~\ref{def:fandg}).
They prove that $\Delta_g$ is a non-negative monotone sub\-modular set function, and thus, 
the well-known greedy algorithm of \cite{nemhauser1978analysis} ensures a $(1-1/e)$ approximation guarantee.

Similarly, the work of Fabbri et al.~\cite{fabbri2022rewiring} considers a gain function in the form of $\Delta_f({\cdot}) = f(\emptyset)-f(\cdot)$ in the context of degree-regular directed graphs and with \emph{rewiring} operations instead of shortcut edge additions. 
Here, the function $f$ is the \emph{largest hitting time} instead of the average (see Definition~\ref{def:fandg}).
Fabbri et al.\ show that the existence of a multiplicative approximation algorithm implies $\mathbf{P} = \NP$.
This result leads them to propose heuristic algorithms. 
However, their inapproximability result does not follow from the use of random walks, 
nor from the rewire operations.
Their reduction still applies if one would replace ``largest hitting time'' 
with ``largest shortest path distance,'' for example. 
Hence, their inapproximability result is an artefact of maximizing a gain function 
in combination with the $\max$ operator in the function $f$ (Definition~\ref{def:fandg}).

\smallskip
The previous discussion motivates the following question:
``\emph{What can we say about directly \emph{minimizing} the functions $f$ and $g$, 
instead of indirectly optimizing them by maximizing associated gain functions?}''

\smallskip
This paper provides several first algorithmic results and ideas regarding this question
for both functions $f$ and $g$. 
We take an abstract view and define our problems on input graphs in their simplest form, 
without the additional constraints (bounded-length random walks or degree-regular graphs) 
considered in \citet{repbub} and \citet{fabbri2022rewiring}, respectively. 
We consider uniform simple random walks of unbounded length on undirected, 
unweighted and connected input graphs.

\ptitle{Results and techniques}
For the first problem we study, 
\emph{minimize average hitting time} (\BMAH), 
we observe that the objective is supermodular 
and it can be optimized by a greedy bi\-criteria strategy~\cite{liberty_et_al:LIPIcs:2017:7568}, 
which offers an $(1+\epsilon)$-approximation at the 
cost of adding logarithm\-ical\-ly more edges. 
In addition, we show how to speed up this algorithm
by deriving an approximation for the objective, 
which relies on sampling a small number of bounded-length absorbing random walks and bounding the error term using eigenvalue techniques.
The length of these walks directly depends on the \emph{mixing time} and \emph{average degree} of the red nodes, which are bounded quantities in most real-life networks.
We then show that running the same greedy strategy of~\citet{liberty_et_al:LIPIcs:2017:7568} 
on the \emph{approximative} values, still provides a $(2+\epsilon)$-approximation~factor.

\smallskip
For our second problem,
\emph{minimize maximum hitting time} (\BMMH), 
the objective function is neither supermodular nor submodular, 
and thus, it is a more challenging task. 
Nevertheless, we present two different algorithms, 
both with approximation guarantees, 
albeit, weaker than in the case of average hitting time.
The first algorithm utilizes a relation between 
maximum and average hitting time, 
and uses the greedy strategy designed for the \BMAH problem.
The second algorithm leverages a novel connection 
between \BMMH and the asymmetric $k$-center problem~\cite{kariv1979algorithmic}, 
and utilizes optimal methods developed for the latter problem~\cite{archer2001two}.

\begin{table}[t]
\caption{Results summary. The approximation factor is computed with respect to an optimal solution that uses at most $k$ edges.
$d_{\text{m}}$ is at most the maximum degree (see Section~\ref{sec:asymm}).}
\label{table:summaryresults}
\centering
\begin{tabular}{ccc}
\hline
& Approx. factor & Number of shortcut edges \\
\hline
\BMAH & $1+\epsilon$  & $\bigO(k \log (n/\epsilon))$  \\
\BMAH & $2+\epsilon$  & $\bigO(k \log (n/\epsilon))$ \\
\BMMH & $\bigO(n^{3/4})$  & $\bigO(k \log n)$ \\
\BMMH & $\bigO\left(\log^*(k)\, d_{\text{m}}\right)$  & $k$ \\
\hline
\end{tabular}
\end{table}

\section{Related work}
\label{sec:rw}

A large body of work has been devoted to designing methods
for optimizing certain graph properties via edge additions, edge rewirings, 
or other graph-edit operations. 
These approaches include methods for 
increasing graph robustness~\cite{chan2016optimizing},
maximizing the centrality of a group of nodes~\cite{medya2018group}, 
improving the betweenness centrality of a node~\cite{bergamini2018improving},
reducing the average shortest path distances over all pairs of 
nodes~\cite{meyerson2009minimizing,parotsidis2015selecting}, 
minimizing the diameter of the graph~\cite{demaine2010minimizing, 10027668}, 
increasing resilience to adversarial attacks~\cite{ma2021graph}, 
and more.
Some of these approaches provide algorithms
with provable approx\-i\-mation guarantees, 
while the others are mostly practical~heuristics.

The setting of the problem we consider has also been studied in the context
of understanding phenomena of bias, polarization, and segregation
in social networks~\cite{bakshy2015exposure,chen2018quantifying,chitra2019understanding,flaxman2016filter,garimella2018quantifying,guerra2013measure,minici2022cascade,ribeiro2020auditing},
and developing computational methods to mitigate those 
adverse effects~\cite{amelkin2019fighting,garimella2017reducing,garimella2017balancing,gionis2013opinion,matakos2020maximizing,musco2018minimizing,tu2020co,zhu2021minimizing,zhu2022nearly}.
Most of the works listed above seek to optimize a complex objective
related to some dynamical diffusion process, 
such as information cascades~\cite{kempe2003maximizing},
or opinion dynamics~\cite{friedkin1990social}.
In all cases, the optimization functions studied in previous work 
are significantly different 
from the hitting-time objectives considered in this paper.

As discussed in the introduction, 
the work that is most closely related to ours,
is the paper by \citet{repbub}, 
which seeks to reduce the structural bias between two polarized groups
of a network via edge insertions. 
Their approach is based on hitting time for bounded-length random walks, 
but contrary to our proposal, 
they seek to \emph{maximize the reduction} of hitting time caused by edge insertions, 
rather than \emph{directly minimizing the hitting time}. 
Similar ideas have been proposed by other authors, 
e.g., \citet{fabbri2022rewiring} and more recenty \citet{coupette2023reducing}, 
in the context of reducing exposure to harmful content. 
We consider our work to be a compelling extension of such previous ideas, 
not only due to optimizing a more natural objective, 
but also for coping with a more challenging optimization problem, 
for which our solution offers new technical ideas and insights. 

\section{Preliminaries}
\label{sec:notations}


\pttitle{Graphs}
We consider undirected connected graphs $G=(V,E)$ with $|V|=n$ nodes, where the nodes $V$ are bi\-partitioned into two groups $R$ and $B$. These groups as referred to as the \emph{red} and \emph{blue} nodes, respectively. We only consider bi\-partitions $V = \{R,B\}$ that are \emph{valid}, meaning that $R$ and $B$ are disjoint and non-empty. For a non-valid bi\-partition, our problem statements (Section~\ref{subs:problem}) are either trivial or ill-defined. An edge is an \emph{inter-group} edge if it has one blue and one red endpoint. 
For a subset of nodes $X \subseteq V$, 
we let $E[X]$ be the subset of edges of $G$ that have both endpoints in $X$, 
and $G[X]=(X,E[X])$ be the subgraph of $G$ \emph{induced} by $X$.
If $F$ is a set of non-edges of $G=(V,E)$, then $G+F$ denotes the graph $(V,E \cup F)$.
The degree of a node $v$ in $G$, which is the number of neighbors of $v$ in $G$, 
is denoted as~$d_G(v)$. 
We let $d_R$ be the average of the degrees of red nodes.

\pttitle{Random walks}
Every random walk will be a uniform simple random walk, unless stated otherwise.
By uniform simple random walk we mean that at each step one of the neighbors of the
current node is selected with uniform probability. 
Given a graph $G=(V,E)$ and $A \subseteq V$, the random variable $\tau_u(A)$ indicates the first time that a random walk on $G$ starting from $u$ visits $A$.
If $u \in A$, we set $\tau_u(A)=0$.
The expectation of this random variable is denoted as $H_G(u,A) = \E_{G}[\tau_u(A)]$, 
where the subscript emphasizes that the random walk is on~$G$.
With a slight abuse of notation, $H_G(u,v)$ is thus the (expected) \emph{hitting time} for 
a walk starting from $u$ to visit $v$ for the first time.
Note that $H_G(u,u) = 0$ and $H_G(u,A) = 0$ if and only if $u \in A$.
It also holds that $H_G(u,A) \leq H_G(u,B)$ 
when $B \subseteq A$, and in particular $H_G(u,A) \leq \min_{a \in A} H_G(u,a)$.

\pttitle{Matrices and norms}
We denote $\mathbf{A}$ as the adjacency matrix of a graph, and $\mathbf{D}$ is the degree diagonal matrix.
Let $e_r \in \mathbb{R}^{r \times 1}$ be the vector of all-ones of size $r$,
and let $\norm{\cdot}_2$ denote the spectral norm.


\pttitle{Additional definitions}
For a graph $G = (V, E)$ with a valid bi\-partition $V = \{R,B\}$,
let $(R \times B) \setminus E$ be the set of inter-group non-edges in $G$. Define the following two set functions $f, g: 2^{(R \times B) \setminus E} \rightarrow \mathbb{R}_{>0}$ 
as the \emph{maximum} and \emph{average hitting time}.

\begin{defi}[Maximum and average hitting time]
\label{def:fandg}
For a set of inter-group non-edges $F \subseteq (R \times B) \setminus E$, 
we define the \emph{maximum hitting time from $R$ to $B$ on graph $G+F$} by
\begin{align*}
f(F) &\triangleq \max_{r \in R} H_{G + F}(r, B),
\end{align*}
and the \emph{average hitting time from $R$ to $B$ on graph $G+F$} by
\begin{align*}
g(F) &\triangleq \frac{1}{|R|}\sum_{r \in R} H_{G + F}(r, B).
\end{align*}
\end{defi}

\section{Problems and observations}
\label{sec:problems-observations}

In this section, we formally define the problems we study,
and we discuss their properties.

\subsection{\label{subs:problem}Problems}


We introduce and study the following two problems, which we denote as the \BMMH (Budgeted Minimum Maximum Hitting time) and \BMAH (Budgeted Minimum Average Hitting time) problems.
\begin{prob}[Budgeted Minimum Maximum Hitting time problem (\BMMH)]
\label{prob:maxht}
Given an undirected connected graph $G = (V, E)$, with a valid bi\-partition $V = \{R,B\}$, 
and a budget $k\in \mathbb{N}$,
we seek to find a set of inter-group non-edges 
$F \subseteq (R \times B) \setminus E$ with $|F| \leq k$ that minimizes $f(F)$, 
where $f$ is the maximum hitting-time function defined in Definition~\ref{def:fandg}.
\end{prob}
\begin{prob}[Budgeted Minimum Average Hitting time problem (\BMAH)]
\label{prob:avght}
Identical definition to Problem~\ref{prob:maxht}, 
but we aim to minimize the function $g$, instead of $f$,
where $g$ is the average hitting-time function defined in Definition~\ref{def:fandg}.
\end{prob}

\BMMH aims to augment $G$ with $k$ new inter-group edges, 
so as to minimize the \emph{largest} hitting time to a blue node, for random walks starting from red nodes.
Likewise, \BMAH aims to minimize the \emph{average} hitting time from red to blue nodes.
Note that for every red node we allow multiple new edges to be incident to it, 
as long as the augmented graph remains simple.

\subsection{Observations}
The \NP -hardness of both \BMMH and \BMAH follows from a reduction from the minimum set cover problem, very similar to the reduction detailed by \citet{repbub} and omitted for brevity. We present several useful observations regarding \BMMH and \BMAH that will guide the design of approximation algorithms:
\begin{obs}
\label{obs:endpoints}
Changing the blue endpoints of the edges in a feasible solution $F$ does not change $f(F)$ or $g(F)$.
\end{obs}
\begin{proof}
Every walk starts from a red node and halts the moment a blue node is hit, 
hence the identity of the blue endpoints of edges in $F$ do not matter.
\end{proof}

Both problems thus reduce to selecting a multiset of red endpoints, 
as we allow multiple edges in $F$ to have the same red endpoint. 
If a red node is connected to all the blue nodes already, 
we can of course not add any more new edges incident to it.

\begin{obs}[monotonicity]
\label{obs:mono}
If $Y \subseteq X$, then $f(X) \leq f(Y)$ and $g(X) \leq g(Y)$.
\end{obs}
\begin{proof}
Let $F$ be a feasible solution.
For every $r \in R$, the hitting time $H_{G + F}(r, B)$ is a monotonically decreasing set function of $F$, 
since the edges in $F$ are inter-group edges. 
Every $e \in F$ only helps in reducing the hitting time from $r$ to $B$. 
This can be proven more formally by a straightforward coupling argument 
between a walker on $G$ and $G + F$~\cite{lovasz1993random, repbub}.
Since individual hitting times are monotone, the functions $f$ and $g$ are monotone, as well.
\end{proof}

The following super\-modularity property for individual hitting times is less obvious, 
and a variant of it was proven by \citet{repbub} for bounded-length random walks on directed graphs.
We give a different and shorter proof of super\-modularity in our case of unbounded random walks, 
by using a coupling argument.

\begin{obs}[super\-modularity of hitting times]
\label{obs:hitsup}
For all $r \in R$ and for all $e_1 \neq e_2 \in (R \times B) \setminus E$ it holds that 
\begin{equation}\label{eq:superm}
H_{G + e_2}(r, B)-H_{G + \{e_1, e_2\}}(r, B) \leq  H_{G}(r, B)-H_{G + e_1}(r, B).
\end{equation}
\end{obs}

\begin{proof}
Pick any $r \in R$. We assume all walks start from $r$.
Consider the right-hand side of Eq.~(\ref{eq:superm}).
Couple walks on $G + e_1$ to walks on $G$ by simply following the walker on $G$, until the walker hits the red endpoint $r_{1} \in R$ of the edge $e_1$.
If the walker hits $r_{1}$, follow the edge $e_1$ with probability $\frac{1}{d_G(r_1)+1}$.
With probability $\frac{d_G(r_1)}{d_G(r_1)+1}$, keep following the walker on $G$. It is clear that the marginal distribution of this coupled walk is a uniform random walk on $G + e_1$.

As long as $e_1$ is not traversed, the two coupled walks are identical.
If $e_1$ is traversed, the walks are identical up until the time prior to traversing $e_1$, when they are both in $r_1$ for the last time.
On this moment, the walker on $G$ still needs to travel $H_{G}(r_1, B)$ steps on average to hit $B$, while the walker on $G + e_1$ needs 1 more step (conditioned on $e_1$ being traversed).
Hence for all $r \in R$ we have
\begin{align*}
H_{G}(r, & B) -  H_{G + e_1}(r, B) \\
=~ & (H_{G}(r_1, B)-1)\, \text{Pr[traversing $e_1$ when walking on $G + e_1$]}.
\end{align*}
Similarly, we can write for the left-hand side of Eq.~(\ref{eq:superm})
\begin{align*}
& H_{G + e_2}(r, B)- H_{G + \{e_1, e_2\}}(r, B)\\
& =  (H_{G + e_2}(r_1, B)-1) \text{Pr[traversing $e_1$ when walking on $G + \{e_1, e_2\}$]}.
\end{align*}
The result follows from Observation~\ref{obs:mono}, since both $H_{G + e_2}(r_1, B) \leq H_{G}(r_1, B)$ and
\begin{align*}
\text{Pr[traversing  } & e_1 \text{ when walking on $G + \{e_1, e_2\}$]} \\
& \leq~ \text{Pr[traversing $e_1$ when walking on $G + e_1$]}.
\end{align*}%
\end{proof}

Observation~\ref{obs:hitsup} states that individual hitting times are super\-modular, 
hence the objective function $g$ of \BMAH is also super\-modular.
However, the objective function $f$ of \BMMH is not guaranteed to be super\-modular, as a result of the max operator.
Observation~\ref{obs:nonweaks} shows that even the weaker notion of 
weakly-$\alpha$ super\-modularity~\cite{liberty2017greedy} does not hold for $f$.
\begin{obs}
\label{obs:nonweaks}
The function $g$ is super\-modular, but 
the function~$f$ is not weakly-super\-modular for any $\alpha \geq 1$.
\end{obs}
\begin{proof}
The super\-modularity of $g$ is direct, as it is an average of super\-modular functions (Observation~\ref{obs:hitsup}). We prove the second statement of the observation.
A non-negative non-increasing set function $\gamma$ on $[n]$ 
is said to be weakly-$\alpha$ super\-modular with $\alpha \geq 1$ if and only if for all $S, T \subseteq [n]: \gamma(S)-\gamma(S\cup T) \leq \alpha \sum_{i \in T \setminus S} (\gamma(S)-\gamma(S \cup \{i\}))$ \cite[Definition 2]{liberty2017greedy}.
Consider the following instance for \BMMH. The input graph is a path of five nodes. All the nodes are red, except the middle node is blue.
Let $S = \emptyset$ and $T$ be the two edges that connect both red endpoints to the blue center.
Then, $f(S)-f(S\cup T)>0$, but for all  $i \in T \setminus S: f(S)-f(S \cup \{i\})=0$. 
\end{proof}

\section{Algorithms for minimizing\\average hitting time}
\label{sec:bmah}

We first discuss our algorithms for the \BMAH problem,
as we can use them as a building block for the \BMMH problem.

\subsection{A {\greedy} bi\-criteria algorithm\label{ss:greedy}}
\begin{algorithm}[t]
\caption{\label{A:GreedyBMAH}[{\greedy}] bicrit.\ $(1+\epsilon)$-approximation for \BMAH.}
\begin{algorithmic}[1]
\Require $G=(R \cup B, E)$, parameter $k \geq 1$ and error $\epsilon>0$.
\State $F_0 \leftarrow \emptyset$

\While{$(i \leftarrow 1; i \leq \ceil{k \ln{(\frac{n^3}{\epsilon})}}; i \leftarrow i + 1)$}
\State $F_{i} \leftarrow F_{i-1} \cup \argmin_{e \in (R \times B) \setminus (E \cup F_{i-1})} g(F_{i-1}\cup{\{e\}})$.
\EndWhile
\Ensure $F_{i}$.
\end{algorithmic}
\end{algorithm}

According to Observation~\ref{obs:nonweaks} the objective function $g$ of \BMAH is super\-modular. 
It also holds that $g \geq 1$, as it takes at least one step to hit a blue node starting from any red node. There are two main ways for minimizing a non-negative monotone-decreasing super\-modular function in the literature. One approach gives bounds that depend on the \emph{total curvature} \cite{sviridenko2017optimal} 
of the objective function. 
For example, the work of \citet[Section 6.2]{sviridenko2017optimal} details a randomized algorithm with guarantee $1+\frac{c}{(1-c)e}+\frac{\bigO(\epsilon)}{1-c}$, for any $\epsilon>0$ and curvature $c \in [0,1)$. However, we do not utilize this algorithm since it is not straightforward how to bound the curvature of $g$, 
nor how to practically implement their algorithm.

Instead, we follow the bi\-criteria approach of \citet{liberty_et_al:LIPIcs:2017:7568}, 
where they show that the classic greedy algorithm
gives a bi\-criteria $(1+\epsilon)$-approximation by allowing the addition of more than $k$ elements.
Their results are applicable if either the objective function is bounded away from zero 
(which holds in our case, since $g \geq 1$), or one has an approximate initial solution. 
An immediate consequence of their results is Lemma~\ref{lem:perA1}.

\begin{lemma}
\label{lem:perA1}
The \greedy strategy (Algorithm~\ref{A:GreedyBMAH}) 
returns a $(1+\epsilon)$-approximation to an optimal solution of \BMAH 
(that adds at most $k$ edges) 
by adding at most $\ceil{k \ln{(\frac{n^3}{\epsilon})}}$ edges.
\end{lemma}

\begin{proof}
We may assume $k \geq 1$.
Let $F^*$ denote an optimal solution to \BMAH with at most $k$ edges.
We always have a crude approximate initial solution $F_0 = \emptyset$, since for any $F^*$ it holds that $g(\emptyset) \leq n^3 \leq n^3g(F^*)$. The first inequality holds because hitting times in connected graphs are upper bounded by $n^3$~\cite{lovasz1993random}. 
The second inequality follows from $g \geq 1$. 
The lemma essentially follows by application of the 
results of \citet[Theorem 6]{liberty_et_al:LIPIcs:2017:7568}. We give a short proof here for completeness. Let $F_{\tau}$ be the output of Algorithm~\ref{A:GreedyBMAH}, where $\tau = \ceil{k \ln{(\frac{n^3}{\epsilon})}}$. By using the super\-modularity and monotonicity of $g$, we get the following recursion: 
\begin{align*}
g(F_{\tau}) - g(F^*) &\leq \left(g(\emptyset) - g(F^*)\right)\left(1-1/k\right)^{\tau} \\
&\leq (n^3 - 1)g(F^*)\left(1-1/k\right)^{k \ln{(n^3/\epsilon)}} \\
&\leq (n^3 - 1)g(F^*)e^{-\ln{(n^3/\epsilon)}} \leq \epsilon g(F^*).
\end{align*}
\end{proof}

Note that finding the new edge that minimizes line~3 in Algorithm~\ref{A:GreedyBMAH} does not require a search over $\bigO(n^2)$ possible non-edges. By Observation~\ref{obs:endpoints} we only need to choose a red endpoint, reducing the search to $\bigO(n)$ possible choices in each iteration.

\subsection{{\greedyplus}: Speeding up {\greedy}\label{ss:greedyp}}
\label{ss:spg}
In practice, lazy evaluation~\cite{minoux2005accelerated} will typically reduce the number of function evaluations that Algorithm~\ref{A:GreedyBMAH} makes to $g$. 
Nonetheless, function evaluations are still expensive. 
They require  either solving a linear system of equations with $\bigO(n)$ variables, 
or directly computing the fundamental matrix $(\mathbf{I}-\mathbf{Q})^{-1}$ 
of an absorbing Markov chain with absorbing states $B$, 
where $\mathbf{Q}$ denotes the probability transition matrix between the transient red 
nodes~\cite{kemeny1983finite}.

\subsubsection{$(1 \pm \epsilon)$-estimation of $g$}
We show how to get a fast estimate~$\hat{g}$ that satisfies w.h.p.\
$\hat{g}(\cdot) \in (1 \pm \epsilon)g(\cdot)$,
for some error~$\epsilon >0$.

The high-level idea is to write $g$ as a certain quadratic form of the fundamental matrix of an absorbing Markov chain with absorbing states being the blue nodes. 
The fundamental matrix of an absorbing Markov chain can be written 
as an infinite matrix power series~\cite{kemeny1983finite}. 
We will estimate this infinite series by a truncated sum.
The truncated sum can be estimated efficiently by simulating random walks of bounded length, 
while the error induced by the truncation can also be shown to be small. 
The point of truncation will depend on certain graph parameters, such as the mixing time and average degree of the red nodes. Since most real-life graphs typically have bounded mixing time and bounded average degree, our approach is expected to work well in practice. 
This idea of truncating a matrix power series has been recently used 
by Peng et al.~\cite{peng2021local} 
and prior to that by Andoni et al.~\cite{andoni2018solving}.

The following holds for any connected graph, so without loss of generality 
we prove our estimates for $g(\emptyset)$.
Let $r = |R|$ throughout this proof. 
Let $h_R \in \mathbb{R}^{r \times 1}$ denote the vector of expected hitting times from the red nodes until absorption by the blue nodes. Let $\mathbf{Q} \in \mathbb{R}^{r \times r}$ be the transition matrix between red nodes, obtained by deleting the rows and columns of corresponding blue nodes in the random-walk transition matrix $\mathbf{P} =  \mathbf{D}^{-1}\mathbf{A}$. 
It is well-known~\cite{kemeny1983finite} that $h_R = (\mathbf{I}_r-\mathbf{Q})^{-1}e_{r}$.
Since $g(\emptyset)$ is the average of $h_R$, we can~write
\begin{equation*}
g(\emptyset) = \frac{1}{r}e_{r}^T(\mathbf{I}_r-\mathbf{Q})^{-1}e_{r} = \frac{1}{r}e_{r}^T \displaystyle\sum_{i=0}^{+\infty}\mathbf{Q}^{i}e_{r},
\end{equation*}
where the second equality follows from $\mathbf{Q}^{i} \to \mathbf{0}$, 
as $i \to +\infty$, since $|B| \geq 1$ by our assumption of valid bi\-partition
\cite[Theorem 1.11.1]{kemeny1983finite}.

Now define 
$p_1 \triangleq \frac{1}{r}e_{r}^T \sum_{i=0}^{\ell}\mathbf{Q}^{i}e_{r}$ and 
$p_2 \triangleq \frac{1}{r}e_{r}^T \sum_{i=\ell+1}^{+\infty}\mathbf{Q}^{i}e_{r}$.
We first show how to estimate $p_1$.

\begin{lemma}
\label{lem:p1}
Let $\epsilon, \delta \in (0,1)$.
For every $u \in R$, run $t \geq \frac{\ell^2}{\epsilon^2} \ln (\frac{2n}{\delta})$ absorbing random walks ($B$ are the absorbing nodes) with bounded length~$\ell$ and record how many steps are taken. Take the empirical average over these $t$ trials, for every $u$, and stack the results in a vector~$\hat{h}$. With probability at least $1-\delta$, $\hat{h} \in (1 \pm \epsilon) \sum_{i=0}^{\ell}\mathbf{Q}^{i}e_{r}$.
Hence, also with probability at least $1-\delta$, it holds that $\frac{1}{r}e_{r}^T \hat{h} \in (1 \pm \epsilon) p_1$.
\end{lemma}

\begin{proof}
We can give a probabilistic interpretation to $p_1$.
We argue that it is the average of the hitting times---over all the red nodes as starting points---for absorbing random walks with bounded length $\ell$. 
In other words, the walk halts when either it has taken more than $\ell$ steps, 
or it hits a blue node, whichever comes first.
Consider such a bounded absorbing walk starting from $u \in R$. Let $n_{uv}$ be the random variable counting the number of times the walk visits red node $v \in R$. 
Let $n^i_{uv}$ be $1$ if this walk is in state $v$ after exactly $i$ steps, and $0$ otherwise. 
Clearly $n_{uv} = \sum_{i=0}^\ell n^i_{uv}$.
So the expectation satisfies
\begin{align*}
\E[n_{uv}] & = \sum_{i=0}^\ell \E[n^i_{uv}] 
             = \sum_{i=0}^\ell (1\cdot \mathbf{Q}^i_{uv} + 0 \cdot (1-\mathbf{Q}^i_{uv})) 
           = \sum_{i=0}^\ell \mathbf{Q}^i_{uv}.
\end{align*}
Thus $\sum_{i=0}^{\ell}\mathbf{Q}^{i}e_{r}$ is a vector where the $u$-th entry is the hitting time of an absorbing (with absorbing states $B$) bounded (with length $\ell$) random walk starting from red node $u$. And $p_1$ is the average of this vector.
To find a $(1 \pm \epsilon)$-estimation of $p_1$ it suffices to find a  $(1 \pm \epsilon)$-estimation of each entry of $\sum_{i=0}^{\ell}\mathbf{Q}^{i}e_{r}$.
Now it is straightforward to estimate every $u$-th entry of $\sum_{i=0}^{\ell}\mathbf{Q}^{i}e_{r}$ within a factor of $(1 \pm \epsilon)$.
Run $t$ walks from $u$ and compute the empirical average of their lengths as the estimator.
Such an estimation is exactly the topic of the \texttt{RePBubLik} paper \cite[Lemma 4.3]{repbub}, and the claim in Lemma~\ref{lem:p1} follows easily from Hoeffding's inequality and the union bound.
\end{proof}

Next we show how to bound $p_2$ as a small fraction of $g(\emptyset)$, for a sufficiently large choice of $\ell$.

\begin{lemma}
\label{lem:p2}
If $\ell \geq \frac{\log (d_{R}/\epsilon(1-\lambda))}{\log (1/\lambda)} -1$, then $p_2 \leq \epsilon g(\emptyset)$. 
\end{lemma}

\begin{proof}
Let $\mathbf{D}_R \in \mathbb{R}^{r \times r}$ be a diagonal matrix with diagonal entries being the degrees of the red nodes in $G$. All entries $d_1, \ldots, d_r$ are at least one, since $G$ is assumed to be connected. 
Recall that $d_R = \frac{d_1+\ldots +d_r}{r}$ is the average of these entries.

Let $\mathbf{A}_R \in \mathbb{R}^{r \times r}$ be the adjacency matrix of $G[R]$.
Note that $\mathbf{Q} = \mathbf{D}_R^{-1}\mathbf{A}_R$. Observe that the matrix
$\mathbf{M} = \mathbf{D}_R^{-1/2} \mathbf{A}_R \mathbf{D}_R^{-1/2} = \mathbf{D}_R^{1/2} \mathbf{Q} \mathbf{D}_R^{-1/2}$ is real-valued, symmetric, and similar to $\mathbf{Q}$. So the real-valued eigenvalues $\lambda_1 \geq \lambda_2 \geq \ldots \geq \lambda_r$ of $\mathbf{M}$ are also the eigenvalues of $\mathbf{Q}$ (by similarity) and $\mathbf{M}$ has a orthonormal eigendecomposition $\mathbf{M} = \mathbf{U}\mathbf{\Lambda}\mathbf{U}^T$. Let $\lambda = \max\{|\lambda_1|,|\lambda_r|\}$ be the spectral radius of $\mathbf{M}$ (and of $\mathbf{Q}$).
By symmetry of $\mathbf{M}$, we have that $\norm{\mathbf{M}}_2 = \lambda$. We know $\lambda < 1$, since $\mathbf{Q}$ is a \emph{weakly-chained substochastic matrix},\footnote{This is the case when for all $i$: either the $i$-th row of $\mathbf{Q}$ sums to a value $<1$, or from $i$ it is always possible to reach another $j$ for which every row sums to a value $<1$.} 
\cite[Corollary 2.6]{Azimzadeh_2018}.
Hence,

\begin{align*}
p_2 &= \norm{p_2}_2 \leq \frac{1}{r} \norm{\mathbf{D}_R^{-1/2}e_{r}}_2 \sum_{i=\ell+1}^{+\infty}\norm{\mathbf{M}^{i}}_2 \norm{\mathbf{D}_R^{1/2}e_{r}}_2 \\
&\leq \frac{1}{r} \norm{\mathbf{D}_R^{-1/2}e_{r}}_2 \sum_{i=\ell+1}^{+\infty}\norm{\mathbf{M}}^{i}_2 \norm{\mathbf{D}_R^{1/2}e_{r}}_2 \\
&\leq \frac{\lambda^{\ell+1}}{r(1-\lambda)}\sqrt{\left(d_1 + \ldots + d_r \right)\left(\frac{1}{d_1} + \ldots + \frac{1}{d_r} \right)} \\
&\leq \frac{\lambda^{\ell+1}}{(1-\lambda)}d_R \leq \epsilon \leq \epsilon g(\emptyset).
\end{align*}
The first and second inequality use the similarity between $\mathbf{Q}$ and~$\mathbf{M}$, 
the ortho\-normal eigen\-decomposition of $\mathbf{M}$, and matrix norm properties.
The third inequality follows from $\lambda < 1$.  
The fourth inequality holds as $\sum_i 1/d_i \leq \sum_i d_i$, as $d_i \geq 1$, and 
by plugging in the choice of $\ell \geq \frac{\log (d_{R}/\epsilon(1-\lambda))}{\log (1/\lambda)} -1$. 
The final step follows from $g(\emptyset) \geq 1$.
\end{proof}

We combine Lemma~\ref{lem:p1} and Lemma~\ref{lem:p2} to obtain our final result.

\begin{lemma}
\label{lem:p3}
Given an undirected connected graph $G = (V, E)$, with a valid bi\-partition $V = \{R,B\}$,
we can find an estimate $\hat{g}(\emptyset) \in (1 \pm \epsilon)g(\emptyset)$ 
with probability at least $1-1/{n^c}$ (for some $c \geq 1$) in time
\begin{equation}
\bigO \left(\frac{c \, n \, \ln(n) \, \log^3 (d_{R}/\epsilon(1-\lambda))}{\epsilon^2 \, \log^3 (1/\lambda)} \right).
\end{equation}
\end{lemma}

\begin{proof}
We use the estimate $\hat{g} = \frac{1}{r}e_{r}^T \hat{h}$ from Lemma~\ref{lem:p1} as our estimate of $g(\emptyset)$ with parameter $\epsilon' = \epsilon/2$. Recall that $g(\emptyset) = p_1 + p_2$. Since $\hat{g}$ is an estimation of $p_1$ by Lemma~\ref{lem:p1}, it holds (w.h.p.) that $\hat{g} \leq (1+\epsilon/2)p_1 \leq (1+\epsilon)g(\emptyset)$ since $p_2 \geq 0$.
On the other hand, $\hat{g} \geq (1-\epsilon/2)p_1 = (1-\epsilon/2)(g(\emptyset)-p_2) \geq (1-\epsilon/2)(1-\epsilon/2)g(\emptyset) \geq (1-\epsilon)g(\emptyset)$, where the second to last inequality uses Lemma~\ref{lem:p2}.

The time complexity follows from Lemma~\ref{lem:p1} by plugging in $\delta = 1/{n^c}$ and the choice of $\ell$ from Lemma~\ref{lem:p2}. We need to do this for every red node and there could be $\bigO(n)$ of them. Every bounded random walk also needs at most $\ell$ steps. 
\end{proof}

\subsubsection{Greedy on the estimated values}
Lemma~\ref{lem:p3} states that we can find good estimates $\hat{g}(\cdot) \in (1 \pm \epsilon) g(\cdot)$ relatively fast and with good probability. A natural question to ask is whether running the greedy bi\-criteria approach from Algorithm~\ref{A:GreedyBMAH} on $\hat{g}$ (instead of~$g$) still ensures some form of approximation guarantee, given that~$\hat{g}$ is not guaranteed to be super\-modular.
A similar question has been answered affirmatively for the classic greedy algorithm of \citet{nemhauser1978analysis} for maximizing a nonnegative monotone sub\-modular function under a cardinality constraint, 
by the work of \citet{Horel}. 
They showed that when $\epsilon < 1/k$, 
the greedy still gives a constant-factor approximation~\cite[Theorem 5]{Horel}.

We show an analogue result for the bi\-criteria greedy approach~\cite{liberty2017greedy} 
to minimization of our monotone decreasing nonnegative super\-modular objective function $g$, under a cardinality constraint. 

\begin{theorem}
Algorithm~\ref{algoG+} returns a set $F_\tau$ that satisfies
\begin{equation}
g(F_\tau) \leq (2+\epsilon) \text{\normalfont OPT},
\end{equation}
where $\text{ \normalfont OPT} = g(O)$ and $O$ is an optimal solution of \BMAH with at most $k$ edges.
\end{theorem}
\begin{proof}
By non-negativity, super\-modularity and monotonicty of~$g$, 
we can write for iteration $i$ in Algorithm~\ref{algoG+}
\begin{align*}
g(F_i)  & - \text{OPT} \leq g(F_i) - g(F_i \cup O) \leq \sum_{e \in O} g(F_i) - g(F_i \cup \{e\})  \\
 & \leq \sum_{e \in O} g(F_i) - \frac{1}{1+\epsilon}\hat{g}(F_i \cup \{e\})
   \leq \sum_{e \in O} g(F_i) - \frac{1}{1+\epsilon}\hat{g}(F_{i+1}) \\
 & \leq \sum_{e \in O} g(F_i) - \frac{1-\epsilon}{1+\epsilon}g(F_{i+1})
   \leq k \left(g(F_i) - \frac{1-\epsilon}{1+\epsilon}g(F_{i+1})\right).
\end{align*}
In the third and fifth inequality we made use of $\hat{g}(\cdot) \in (1 \pm \epsilon) g(\cdot)$, and the fourth inequality utilizes the greedy choice in Algorithm~\ref{algoG+}. So we have the following recursion
\begin{equation}
g(F_{i+1}) \leq \left(\frac{1+\epsilon}{1-\epsilon}\right)(1-1/k)g(F_i) + \left(\frac{1+\epsilon}{1-\epsilon}\right)\frac{\text{OPT}}{k}.
\end{equation}
This recursion is of the form of $u_{i+1} \leq au_i+b$, which after $\tau$ iterations satisfies $u_{\tau} \leq a^{\tau}u_0 + \frac{1-a^{\tau}}{1-a}b$. Now write $\epsilon = \delta/k$ for $\delta \in (0,1/4]$. 
With our specific choice of $a = \frac{1+\epsilon}{1-\epsilon}(1-1/k)$ and  
$\tau = \ceil{2k \ln{(\frac{n^3}{\epsilon})}}$  
we upper bound $a^{\tau}u_0 = a^{\tau}g(\emptyset)$ as 
\begin{align*}
a^{\tau}g(\emptyset) & = \left(\left(\frac{k+\delta}{k-\delta}\right)(1-1/k)\right)^{{\tau}}g(\emptyset)
\leq e^{(2\delta-1)2\ln(n^3/\epsilon)}g(\emptyset) \\
&\leq e^{-\ln(n^3/\epsilon)}g(\emptyset) = \frac{\epsilon}{n^3}g(\emptyset) \leq \epsilon \text{OPT}.
\end{align*}
where the first inequality follows from Fact~\ref{fact:inequality} and the last inequality from the discussion in the proof of Lemma~\ref{lem:perA1}.

\begin{fact}
\label{fact:inequality}
Let $x \in \mathbb{R}$ with $|x|\leq 1$. Let $k \geq 1$ be a positive integer. Then it holds that
\begin{equation}
\left(\frac{k+x}{k-x} \left(1-\frac{1}{k}\right)\right)^k \leq e^{2x-1}.
\end{equation}
\end{fact}

Next we upper bound the second part $\frac{1-a^{\tau}}{1-a}b$.
It can be verified that $0 \leq a<1$ for parameters $\epsilon = \delta/{k}$,  $\delta \leq 1/4$ and $ k \geq 1$. Thus,
\begin{align*}
\frac{1-a^{\tau}}{1-a}b \leq \frac{1}{1-a}b = \frac{1+\epsilon}{1-2\epsilon k + \epsilon}\,\text{OPT} \leq \frac{1}{1-2\delta}\,\text{OPT} \leq 2\,\text{OPT}.
\end{align*}

\end{proof}
\begin{algorithm}[t]
\caption{\label{algoG+}[{\greedyplus}] bicrit. $(2+\epsilon)$-approximation for \BMAH{}.}
\begin{algorithmic}[1]
\Require $G=(R \cup B, E)$, parameter $k \geq 1$, error $\epsilon \in (0,\frac{1}{4k}]$ and estimates $\hat{g}(\cdot) \in (1 \pm \epsilon) g(\cdot)$ as found by Lemma~\ref{lem:p3}.
\State $F_0 \leftarrow \emptyset$

\While{$(i \leftarrow 1; i \leq \ceil{2k \ln{(\frac{n^3}{\epsilon})}}; i \leftarrow i + 1)$}
\State $F_{i} \leftarrow F_{i-1} \cup \argmin_{e \in (R \times B) \setminus (E \cup F_{i-1})} \hat{g}(F_{i-1}\cup{\{e\}})$.
\EndWhile
\Ensure $F_{i}$.
\end{algorithmic}
\end{algorithm}

\section{Algorithms for minimizing\\maximum hitting time}
\label{sec:bmmh}
Next we turn our attention to the \BMMH problem.
We first show that a $\alpha$-approximation for \BMAH results in an $(2|R|^{3/4}\alpha)$-approx\-i\-ma\-tion 
for \BMMH.
We then give a completely different approach based on 
the asymmetric $k$-center problem~\cite{kariv1979algorithmic}.

\subsection{Relating the average and the maximum hitting times \label{ss:relatingavgm}}
The \BMMH objective function $f$ is not super\-modular (Observation~\ref{obs:nonweaks}), although we know that for every $r \in R$ each individual hitting time $H_{G + F}(r, B)$ is a super\-modular and monotonically decreasing set function of $F$ (Observation~\ref{obs:mono} and Observation~\ref{obs:hitsup}).

We exploit this idea, by upper bounding $f$ by an expression that involves $g$ (Theorem~\ref{thm:boundsF}).
Since also $g \leq f$, we can use any approximation algorithm for minimizing $g$ 
(for instance, the methods in Sections~\ref{ss:greedy} and~\ref{ss:greedyp}) 
and use the upper bound to derive an approximation guarantee for~\BMMH.

Theorem~\ref{thm:boundsF} relates the average and maximum hitting time from red to blue nodes.
We credit Yuval Peres \cite{Peres} for the proof of Theorem~\ref{thm:boundsF},
in case the set $B$ is a singleton, and extend the proof here when $B$ contains more than one node.
Theorem~\ref{thm:boundsF} gives a remarkable non-trivial $\bigO(n^{3/4})$ upper bound on the ratio of the maximum and average hitting times between two groups that constitute a bi\-partition of the graph.
Theorem~\ref{thm:counterboundsF} shows that there exists graphs for which this is tight. In contrast, if we disregard the bi\-partition, and take the average and the maximum hitting time over \emph{all nodepairs} instead, then it is known that this ratio can grow as bad as $\bigO(n)$, but not worse \cite{CSPOST}.
\begin{theorem}
\label{thm:boundsF}
Given an undirected connected graph $G$ with a valid bipartation $V=\{R,B\}$. Let $F \subseteq (R \times B) \setminus E$ be a set of non-edges. Then,
\begin{equation}\label{eq:boundsF}
g(F) \leq f(F) \leq 2|R|^{3/4}g(F).
\end{equation}
\end{theorem}

\begin{proof}
It clearly suffices to prove this for $F = \emptyset$.
Since $G$ is connected, the hitting time between any two nodes is finite-valued.
The first inequality is immediate since the maximum is not less than the average.
Consider the second inequality. Let $\alpha>0$ be some real number, which we will later choose carefully. Categorize the nodes into two groups, the ones with a ``high'' hitting time and the ones with a ``low'' hitting time to the blue nodes. More concretely, define
\begin{equation*}
S=\{v \in V: H_{G}(v,B) \leq \alpha g(\emptyset)\},
\end{equation*}
and $S^c$ as the complement of this set.
Since
\begin{equation*}
|R| g(\emptyset) = \sum_{r \in R} H_{G}(r, B) \geq \sum_{r \in S^c} H_{G}(r, B) > |S^c|\alpha g(\emptyset),
\end{equation*}
it follows that $|S^c| < |R|/\alpha$. The first inequality holds because $S^c \subseteq R$ (since $B \subseteq S$), and the second inequality uses the definition of $S^c$.

Now consider the graph $G^*$ with $|S^c|+1$ nodes, obtained by the following contraction in $G$: 
All nodes in $S$ are contracted to a single supernode $s$. If $w \in S^c$ has several edges connecting it to $S$ in $G$, we only keep one of these edges. This edge connects $w$ to $s$ in $G^*$.
This contraction process has the property that for all $v \in S^c$ it holds 
$H_G(v,S) \leq H_{G^*}(v,s)$, which is easily proved by a straightforward coupling argument.
Now we can bound $H_{G^*}(v,s)$. 
Since $G^*$ is connected, contains no multi-edges, and has $|S^c|+1$ nodes, it follows that $H_{G^*}(v,s) \leq |S^c|^3$ by a well-known result (see, for instance, \citet{LAWLER198685}, and \citet{Wrinkler}).
By the strong Markov property \cite[Appendix A.3]{levin2017markov}, it holds for all $v \in V:$
\begin{align}\label{lasteq:prop}
H_G(v,B) &\leq H_G(v, S) + \max_{u \in S} H_G(u,B) \notag\\
&\leq |S^c|^3 + \alpha g(\emptyset) < \big((|R|/{\alpha})^3 + \alpha \big)g(\emptyset),
\end{align}
where the last inequality holds because $|S^c| < |R|/\alpha$ and $g(\emptyset) \geq 1$.
Now Eq.~(\ref{lasteq:prop}) is minimized for the choice $\alpha = |R|^{3/4}$, and hence $f(\emptyset) \leq 2|R|^{3/4}g(\emptyset)$.
\end{proof}

To see the tightness of the upper bound in Theorem~\ref{thm:boundsF}, consider the following family of graphs.
\begin{theorem}
\label{thm:counterboundsF}
There exists graphs for which the ratio $f(\emptyset)/g(\emptyset) \\ \in \Omega(n^{3/4})$.
\end{theorem}
\begin{proof}
Consider the following star-path-clique configuration.
We have a star consisting of $n$ nodes, with center node $x$.
Attach to $x$ a $(n^{1/4},n^{1/4})$-lollipop graph (see \citet{Wrinkler}).
More precisely, attach to $x$ a path of $n^{1/4}$ nodes and to the other endpoint of this path we attach a clique of also $n^{1/4}$ nodes. Set $B = \{x\}$ and $R = V \setminus B$.
Clearly $f(\emptyset) \in \Theta(n^{3/4})$ since a random walk starting from a clique node needs a cubic number of expected steps with respect to the size of the lollipop to reach $x$ \cite{Wrinkler}. On the other hand,
\begin{equation*}
g(\emptyset) \in \bigO\left(\frac{n\cdot1+n^{1/4}\cdot n^{3/4}}{n}\right) = \bigO(1),
\end{equation*}
and the claim follows.
\end{proof}

Putting the pieces together, we get the following result.
\begin{cor}
\label{cor:approxgf}
An $\alpha$-approximation for \BMAH is a $(2|R|^{3/4}\alpha)$-approximation for \BMMH.
\end{cor}

\begin{proof}
Consider the output $F$ of $\alpha$-approximation algorithm for \BMAH.
Let $O_g$ (resp. $O_f$) be an optimal solution of \BMAH (resp. \BMMH).
Using Theorem~\ref{thm:boundsF} it holds that $f(F) \leq 2|R|^{3/4}g(F) \leq 2|R|^{3/4}\alpha g(O_g) \leq 2|R|^{3/4}\alpha g(O_f) \leq 2|R|^{3/4}\alpha f(O_f)$.
\end{proof}

\subsection{An asymmetric $k$-center approach\label{ss:kcenter}\label{sec:asymm}}

Hitting times in both undirected and directed graphs form a quasi-metric, in the sense that they satisfy all requirements to be a metric except symmetry, since it generally holds that $H_G(u,v) \neq H_G(v,u)$.
However, hitting times do satisfy the triangle inequality \cite{lovasz1993random}.
We will try to exploit this, by linking the \BMMH problem to the following variant of the asymmetric $k$-center problem.

\begin{prob}\label{prob:asymkcent_variant}
[Asymmetric $k$-center Problem with one fixed center] 
Given a quasi-metric space $(V,d)$ and a point $x \in V$.
Find a set of $k$ centers $F \subseteq V$, $|F| \leq k$ that minimize
\begin{align*}
\max_{v \in V} d(v, F \cup \{x\}),
\end{align*}
where for a set $A \subseteq V$ it holds $d(v,A) = \min_{a \in A} d(v,a)$.
\end{prob}
Problem~\ref{prob:asymkcent_variant} is a variant of the classic asymmetric $k$-center problem \cite{kariv1979algorithmic}. 
It is not hard to see that Problem~\ref{prob:asymkcent_variant} is as hard to approximate as the classic problem. Additionally, existing algorithms \cite{archer2001two,panigrahy1998ano} for the classic problem can easily be modified to solve Problem~\ref{prob:asymkcent_variant}.

We show a relationship between \BMMH and Problem~\ref{prob:asymkcent_variant}.
Let $\mathcal{I}_1$ be an instance to \BMMH.
Lemma~\ref{lem:asmkc} states that the optimum of Problem~\ref{prob:asymkcent_variant} on a well-chosen quasi-metric space $\mathcal{I}_2$ (which depends on the instance $\mathcal{I}_1$) is a lower bound for the optimum of \BMMH on instance $\mathcal{I}_1$.
We define this quasi-metric space in Definition~\ref{def:modhtqm}.

\begin{defi}[modified hitting time quasi-metric\label{def:modhtqm}]
Let $G = (R \cup B, E)$ be an input graph to \BMMH.
Associate with $B$ a single new point $b$, and 
consider function $d: (R \cup \{b\}) \times (R \cup \{b\}) \to \mathbb{R}_{\geq 0}$
defined as
\begin{equation}
  d(u,v) \triangleq
    \begin{cases}
      H_G(u,v) & \text{if $u,v \in R$},\\
      H_G(u,B) & \text{if $u \in R$ and $v=b$},\\
      \max_{i \in B} H_G(i,v) & \text{if $u=b$ and $v\in R$},\\
      0 & \text{if $u=v=b$}.
    \end{cases}       
\end{equation}
\end{defi}

\begin{proposition}
$(R \cup \{b\}, d)$ is a quasi-metric space.
\end{proposition}
\begin{proof}
We prove the triangle inequality $d(x,y) \leq d(x,z) + d(z,y)$, as the identity axiom is satisfied.
The triangle inequality holds when $b \neq x,y,z$, because of standard triangle inequality of hitting times.
So consider the case where $b$ is present in the inequality. 
There are three cases involving $u,v \in R$ and $b$:
\begin{itemize}
\item $d(u,v) \leq d(u,b) + d(b,v)$ or equivalently $H_G(u,v) \leq H_G(u,B) + \max_{i \in B} H_G(i,v)$.
This holds because of the strong Markov property when visiting the set $B$.
\item $d(u,b) \leq d(u,v) + d(v,b)$ or equivalently $H_G(u,B) \leq H_G(u,v) + H_G(v,B)$.
This holds because of triangle inequality of hitting times.
\item $d(b,u) \leq d(b,v) + d(v,u)$ or equivalently
$\max_{i \in B} H_G(i,u) \leq \max_{j \in B} H_G(j,v) + H_G(v,u)$.
Let $i^* = \argmax_{i \in B} H_G(i,u)$. 
This holds because $H_G(i^*,u) \leq H_G(i^*,v) + H_G(v,u) \leq \max_{j \in B} H_G(j,v) + H_G(v,u)$.
\end{itemize}
\end{proof}

\begin{lemma}
\label{lem:asmkc}
Let $M^*$ be the optimum of \BMMH for an input graph~$G$.
Let $C^*$ be the optimum of Problem~\ref{prob:asymkcent_variant} 
on the quasi-metric space defined in Definition~\ref{def:modhtqm} with $x=b$.
Then, it holds that
\begin{align*}
C^* \leq M^*.
\end{align*}
\end{lemma}
\begin{proof}
Consider the red endpoints $C = \{c_1,\ldots,c_k\} \subseteq R$ of a set of edges $F$ to \BMMH that achieves the optimum value $M^*$.
Note that some endpoints might be the same, so $C$ is a multiset.
Let $G^* = G+F$ be the augmentation of $G$ with $F$.

Create a feasible solution to Problem~\ref{prob:asymkcent_variant} by taking $C$ as the set of chosen centers. We claim that for every $v \in R \cup \{b\}$ it holds that $d(v, C \cup \{b\}) \leq M^*$.
We have three cases.
If $v = b$, the claim is trivially satisfied.
If $v \neq b$ and $H_G(v,B) \leq M^*$, the claim also holds, as from Definition~\ref{def:modhtqm} we have $d(v, C \cup \{b\}) \leq d(v, b) = H_G(v,B)\leq M^*$.
So the remaining case is when $v \neq b$ and $H_G(v,B) > M^*$.
Since $H_{G^*}(v,B) \leq M^*$ (by definition of $M^*$) there must exist a non-empty set of endpoints $I \subseteq C$, such that for all $i \in I$:
\begin{align*}
\text{Pr}[\text{walk starting from $v$ uses a new edge incident to } i]>0,
\end{align*}
when walking on the graph $G^*$.
Let $I$ be the largest such set.
Using the law of total expectation, decompose\footnote{In Eq.~(\ref{eq:ht}) we have assumed that $\text{Pr}[\text{walk uses no new edges}]>0$. The analysis when this probability is zero is easier and follows the same reasoning.} $H_{G^*}(v,B)$ as
\begin{align}\label{eq:ht}
H_{G^*}(v,B) =~ & H_{G^*}(v,B | \text{ walk uses no new edges}) \\
&\quad\cdot \text{Pr}[\text{walk uses no new edges}]\notag \notag\\
& + \sum_{i \in I} H_{G^*}(v,B | \text{ walk uses a new edge incident to } i) \notag\\
&\quad\cdot \text{Pr}[\text{walk uses a new edge incident to } i] \leq M^*.\notag
\end{align}

By our case assumption $H_{G}(v,B)> M^*$, it holds that
\begin{align*}
H_{G^*}(v,B | \text{ walk uses no new edges})=H_{G}(v,B)> M^*.
\end{align*}

Thus the only way Eq.~(\ref{eq:ht}) can hold is when
\begin{align*}
\min_{i \in I} H_{G^*}(v,B \mid \text{walk uses a new edge incident to } i) \leq M^*. 
\end{align*}

Let $i^*$ be a center $i \in I$ that attains the minimum.
The result follows since
\begin{align*}
d(v, & C \cup \{b\}) \leq d(v, C) = \min_{c \in C} H_G(v, c) \leq H_G(v, i^*)\\
&\leq H_{G^*}(v,B \mid \text{walk uses a new edge incident to } i^*) \leq M^*.
\end{align*}
\end{proof}

\begin{algorithm}[t]
\caption{\label{A:BMMHasym}[\texttt{ASyMM}] $\bigO\left(\beta\, d_{\text{m}}\right)$-approximation for \BMMH{}.}
\begin{algorithmic}[1]
\Require $G=(R \cup B, E)$, parameter $k \geq 1$.
\State Run any $\beta$-approximation algorithm \cite{archer2001two,panigrahy1998ano} for asymmetric $k$-center variant (Problem~\ref{prob:asymkcent_variant}) on the quasi-metric space defined in Definition~\ref{def:modhtqm}, with fixed center $x=b$.
\State Let $C = \{c_1,\ldots,c_k\} \subseteq R$ be the output of step 1. Return a set $F$ of at most $k$ shortcut edges by connecting every $c_i$ to some node in $B$.
\end{algorithmic}
\end{algorithm}

Lemma~\ref{lem:asmkc} allows us to analyze Algorithm~\ref{A:BMMHasym}.
We write $d_{\text{m}} = \max_{c \in C} d_G(c)$ 
for the maximum degree of the nodes in $C$, and $\beta$ 
for the approximation guarantee of an asymmetric $k$-center algorithm.
The state-of-the-art approximation for asymmetric $k$-center is
$\beta = \bigO(\log^*(k))$ by two different algorithms of \citet{archer2001two}.

\begin{theorem}
Algorithm~\ref{A:BMMHasym} is an $\bigO\left(\beta\, d_{\text{m}}\right)$-approximation for \BMMH{}.
\end{theorem}

\begin{proof}
By using a $\beta$-approximation for Problem~\ref{prob:asymkcent_variant} and using Lemma~\ref{lem:asmkc} we get that for all $r \in R$ it is $H_G(r,C) \leq \min_{c \in C} H_G(r,c) \allowbreak \leq \beta C^* \leq \beta M^*$.
Consider a random walk that starts from any $r \in R$. In expectation we need at most $\beta M^*$ steps to encounter a node $c \in C$. We know that every $c \in C$ is incident to \emph{some} edge connecting it to $B$. This is either a new shortcut edge $e \in F$ from Algorithm~\ref{A:BMMHasym}, or if we cannot add a new shortcut edge incident to $c$, this implies that $c$ is already connected to all the blue nodes. In either case, when the walk is at $c$, we have a chance of at least $\frac{1}{d_{\text{m}}+1}$ to hit the blue nodes in the next step. On the other hand, if the walk does not follow such an edge, the walk needs at most another $\beta M^*$ steps in expectation to reach $C$ again, and the process repeats. Of course, the walk might have encountered a blue node along the way, but this is only beneficial. So we conclude that for all $r \in R$:
\begin{align*}
H_{G + F}(r,B)  & \leq  (\beta M^*+1)\frac{1}{d_{\text{m}}+1} + (2 \beta M^*+1)\frac{d_{\text{m}}}{d_{\text{m}+1}}\frac{1}{d_{\text{m}}+1} + \\
&(3 \beta M^*+1)\left(\frac{d_{\text{m}}}{d_{\text{m}+1}}\right)^2 \frac{1}{d_{\text{m}}+1} + \ldots = 1+(d_{\text{m}}+1)\beta M^*.
\end{align*}
The last step follows by splitting the sum as a geometric series and an arithmetico-geometric series, which states that $\sum_{i=1}^{\infty}ir^i = \frac{r}{(1-r)^2}$ for $0<r<1$.
\end{proof}


\begin{figure*}[t]
\centering
\subfloat{\scalebox{0.23}{\includegraphics[trim={0 0 1.53cm 1cm},clip]{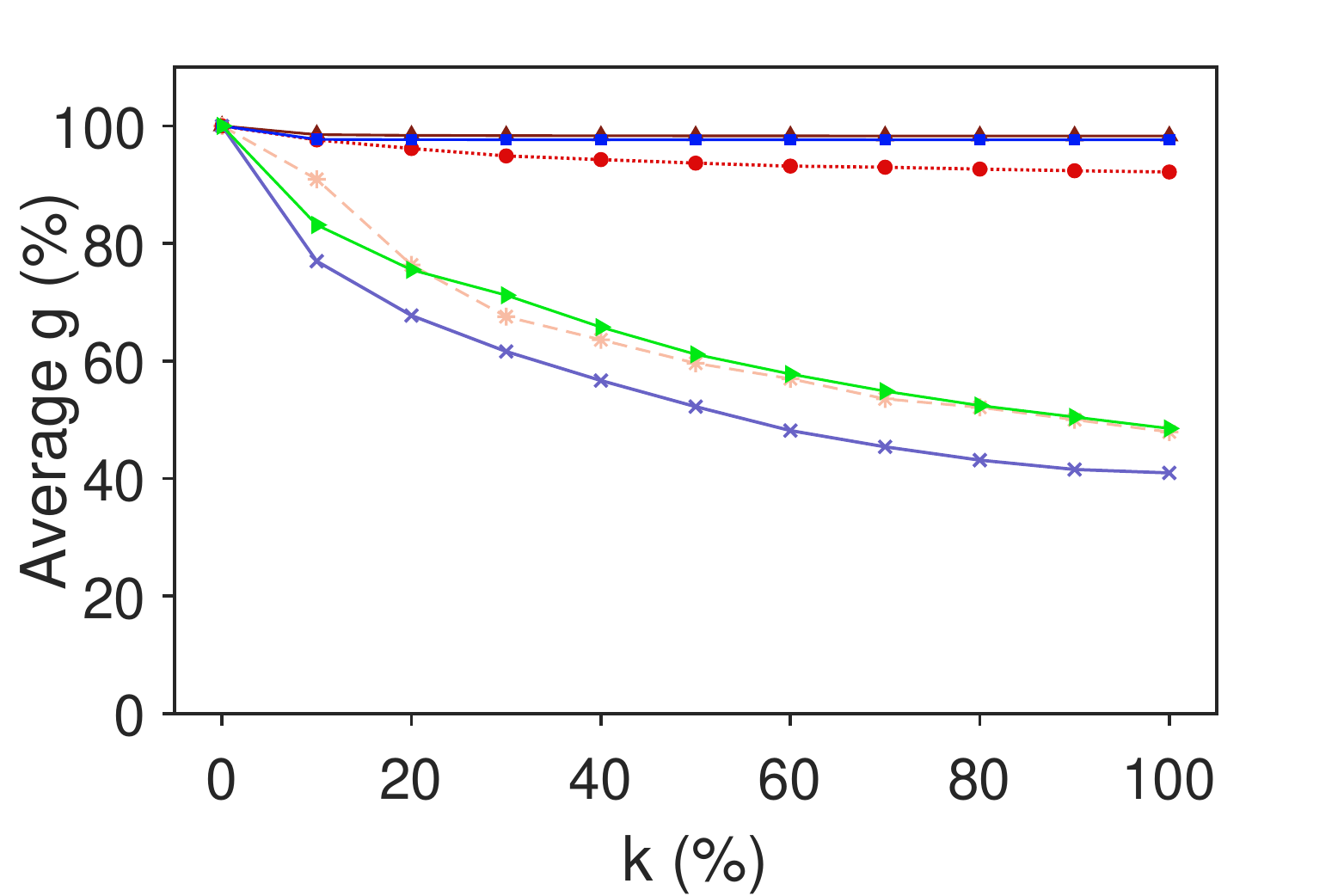}}}
\subfloat{\scalebox{0.23}{\includegraphics[trim={0 0 1.53cm 1cm},clip]{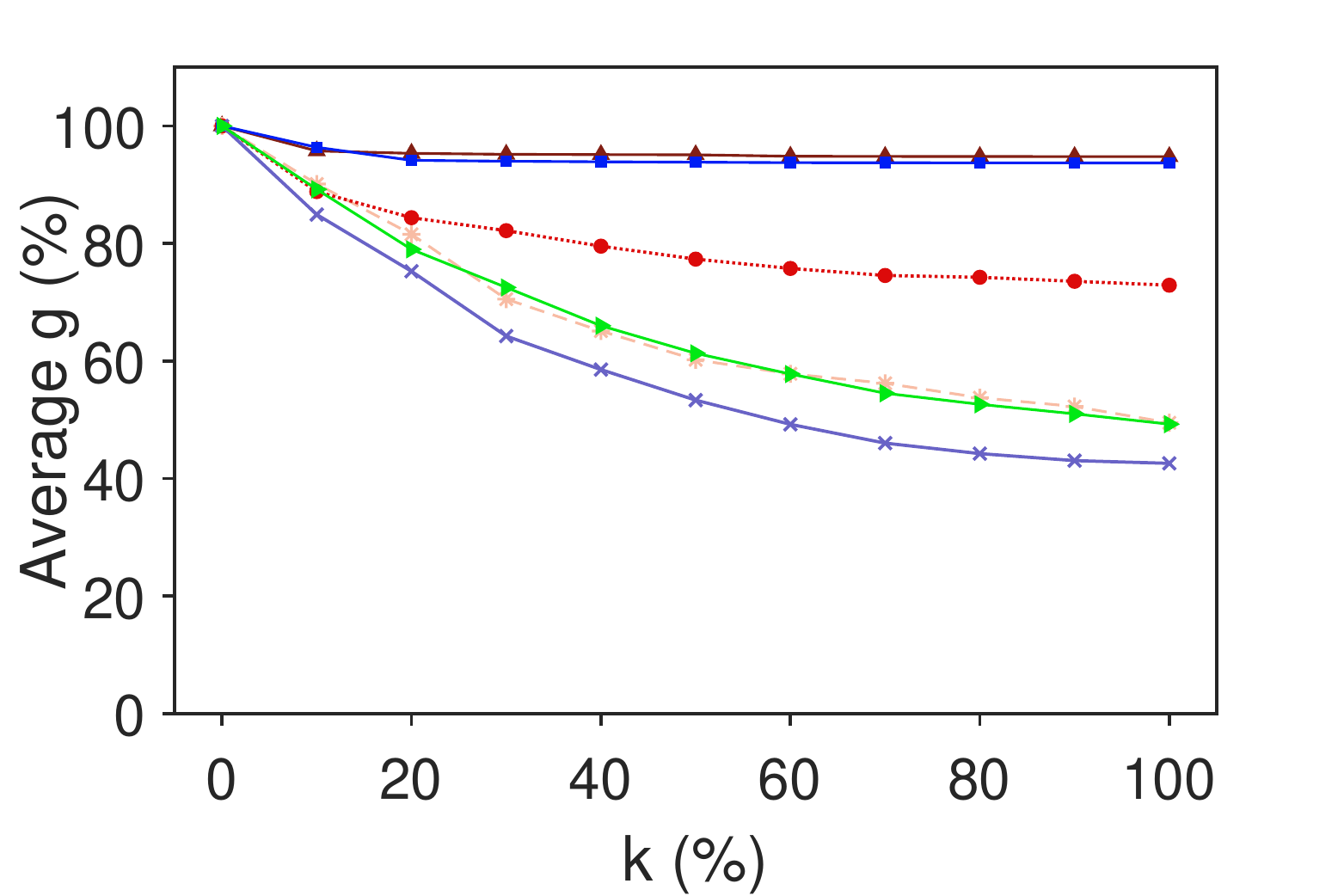}}}
\subfloat{\scalebox{0.23}{\includegraphics[trim={0 0 1.53cm 1cm},clip]{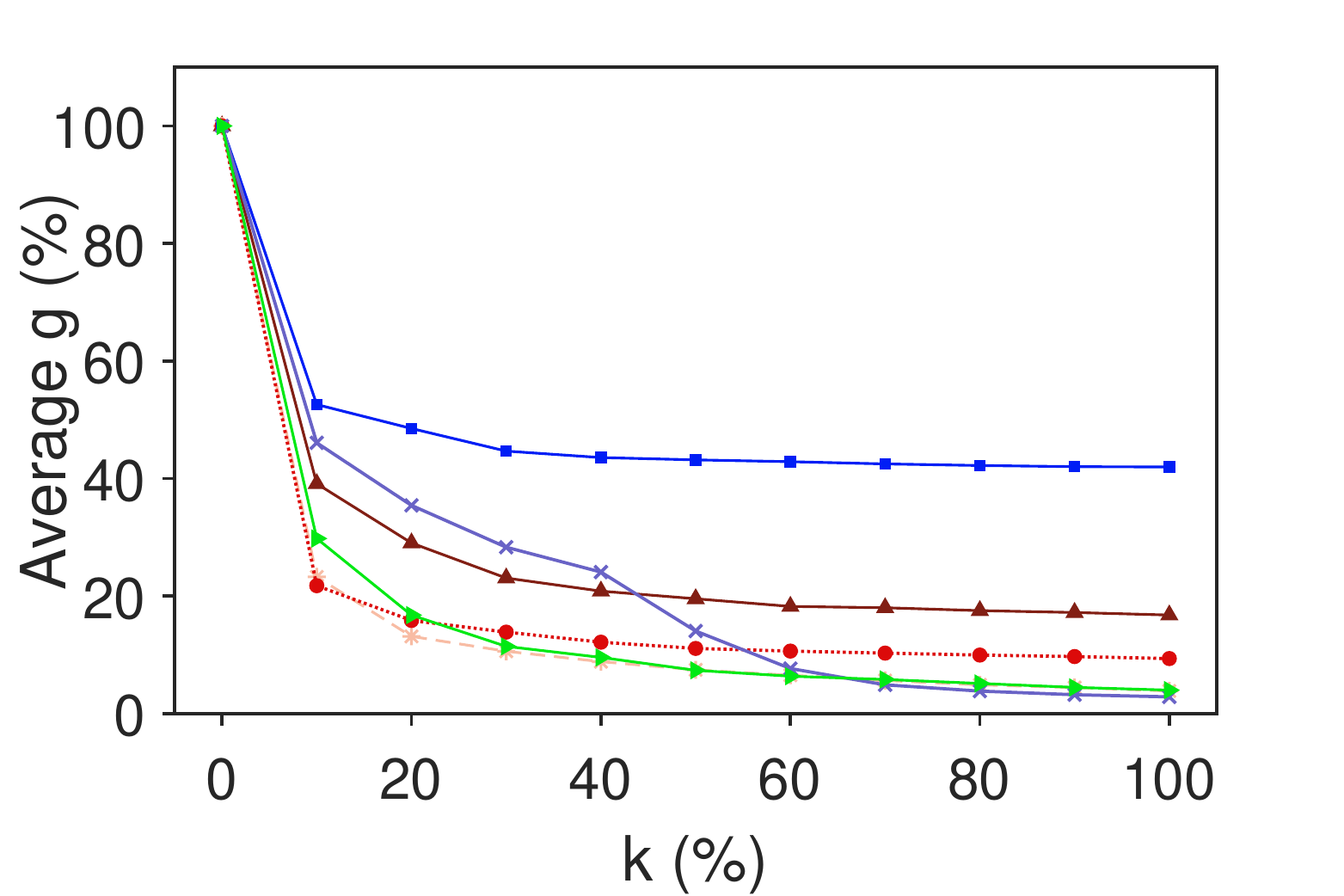}}}
\subfloat{\scalebox{0.23}{\includegraphics[trim={0 0 1.57cm 1cm},clip]{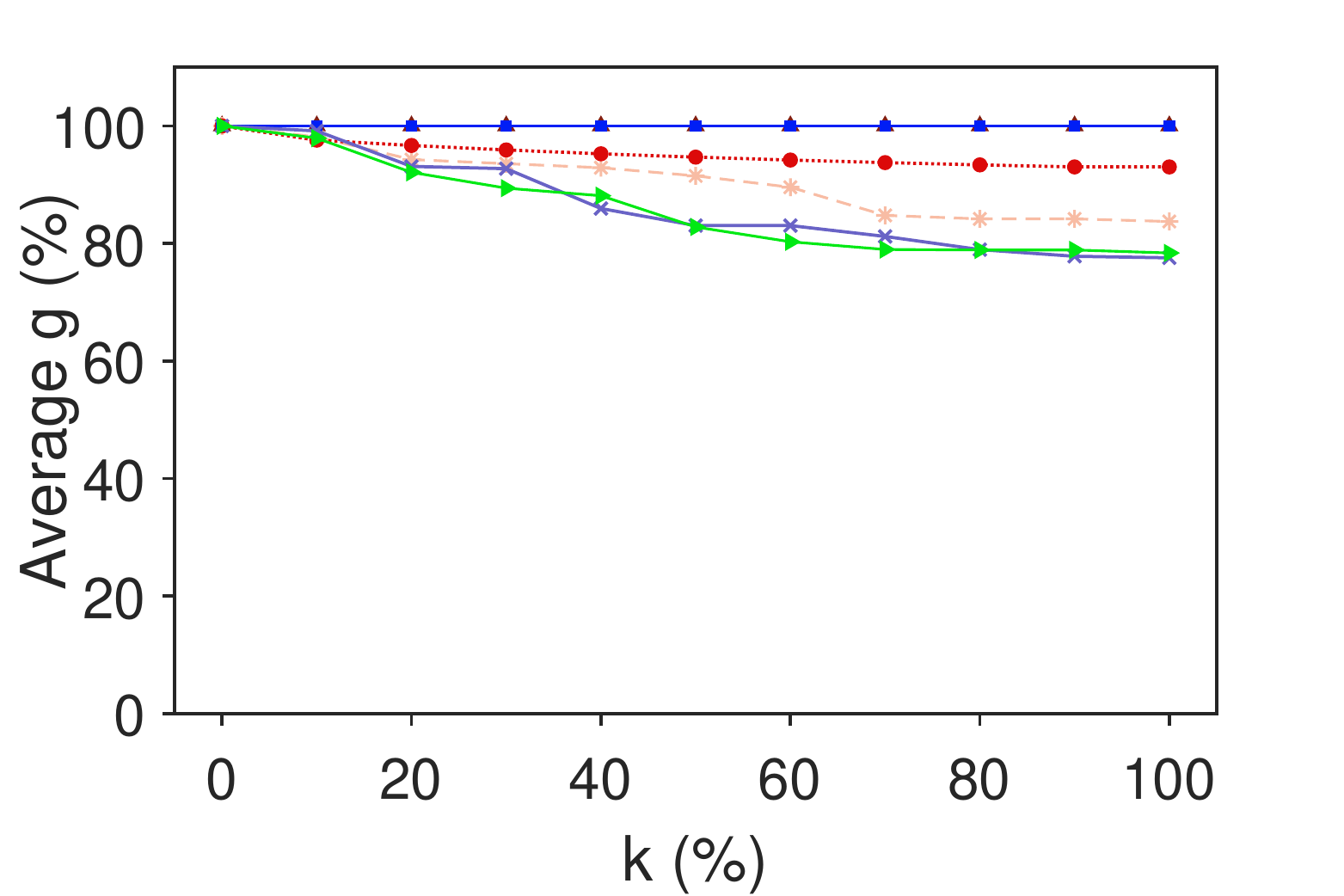}}}
\subfloat{\scalebox{0.23}{\includegraphics[trim={0 0 1.57cm 1cm},clip]{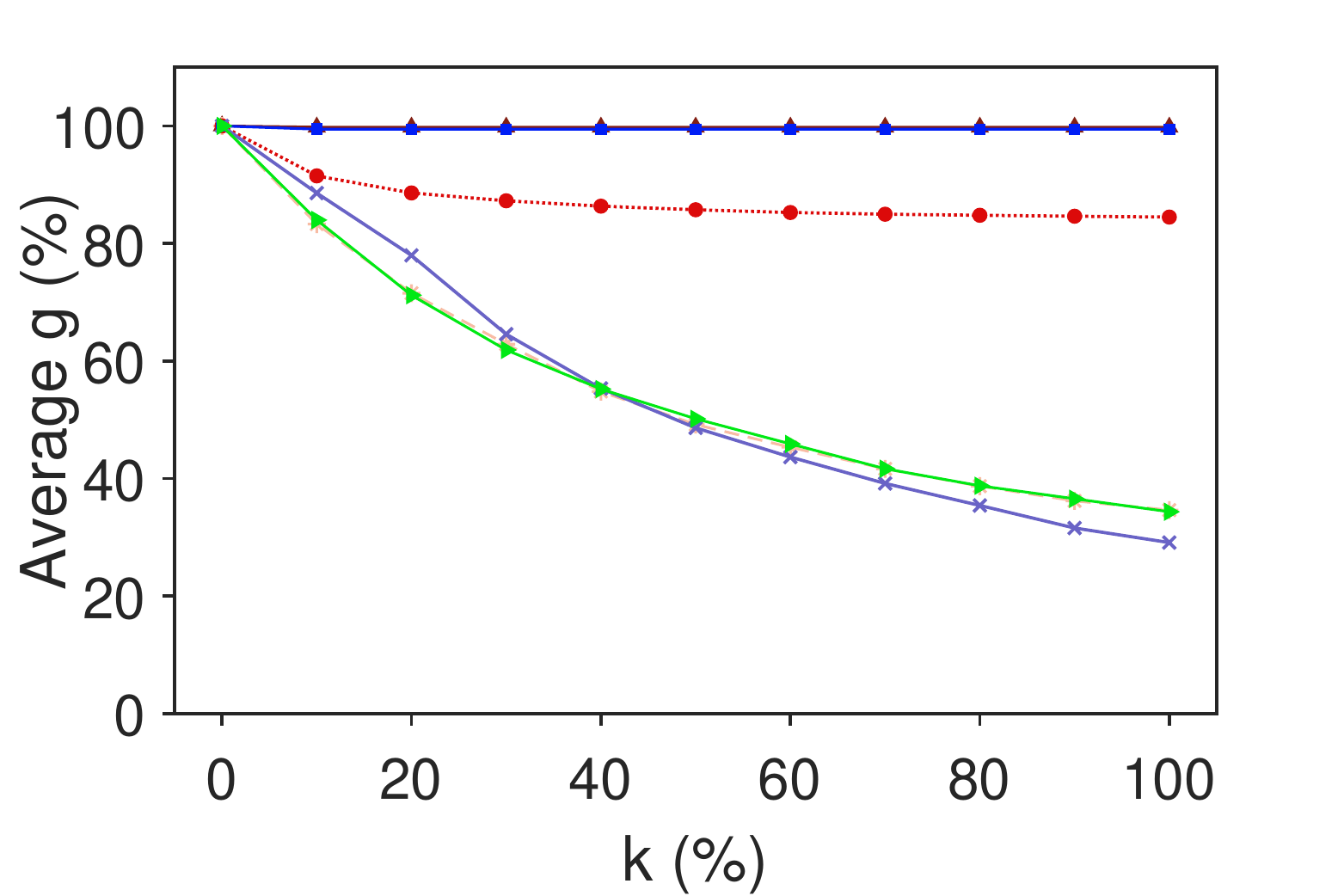}}}
\\
\renewcommand{\thesubfigure}{a}
\subfloat[Wiki Abort.]{\scalebox{0.23}{\includegraphics[trim={0 0 1.53cm 1cm},clip]{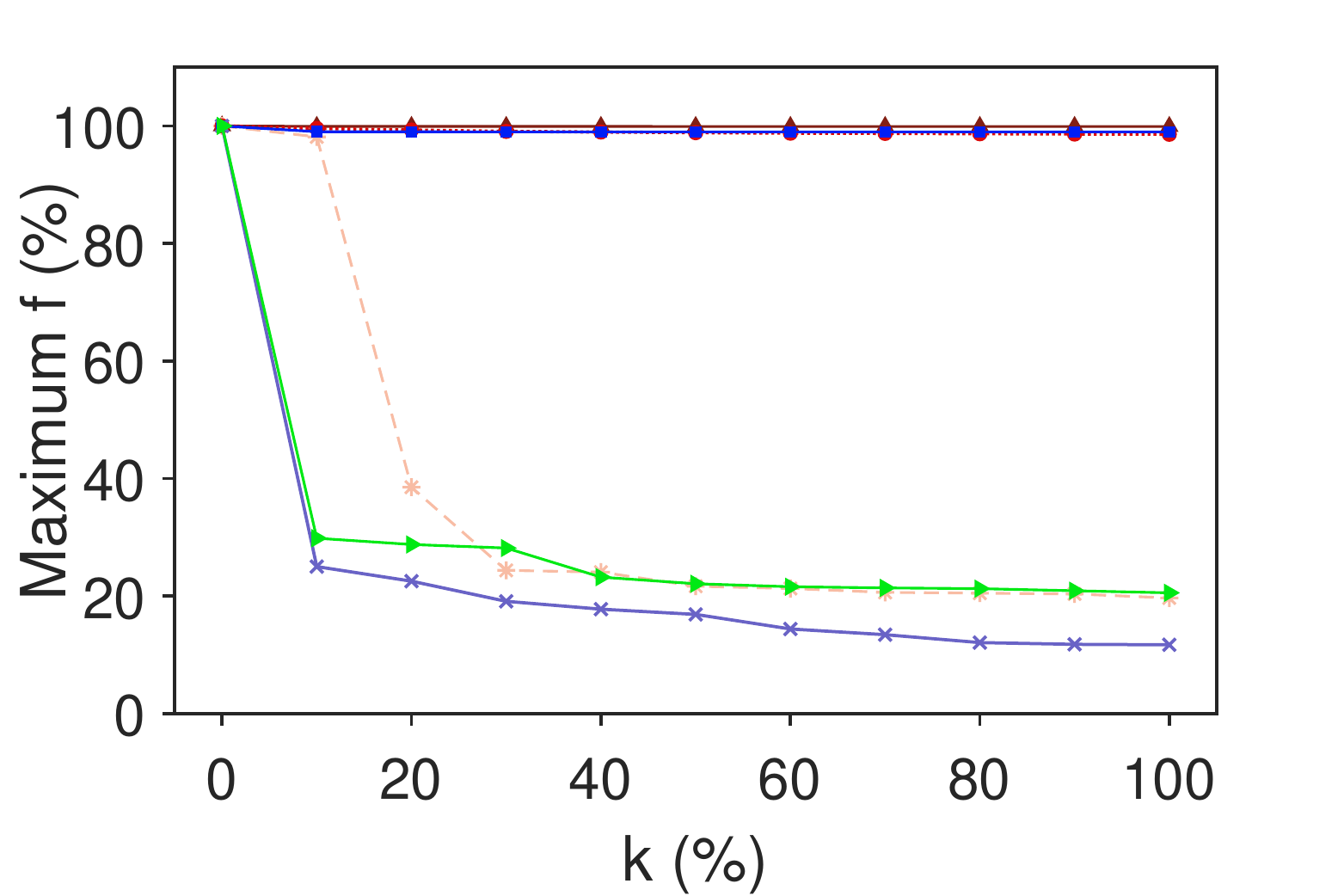}}}
\renewcommand{\thesubfigure}{b}
\subfloat[Wiki Guns]{\scalebox{0.23}{\includegraphics[trim={0 0 1.53cm 1cm},clip]{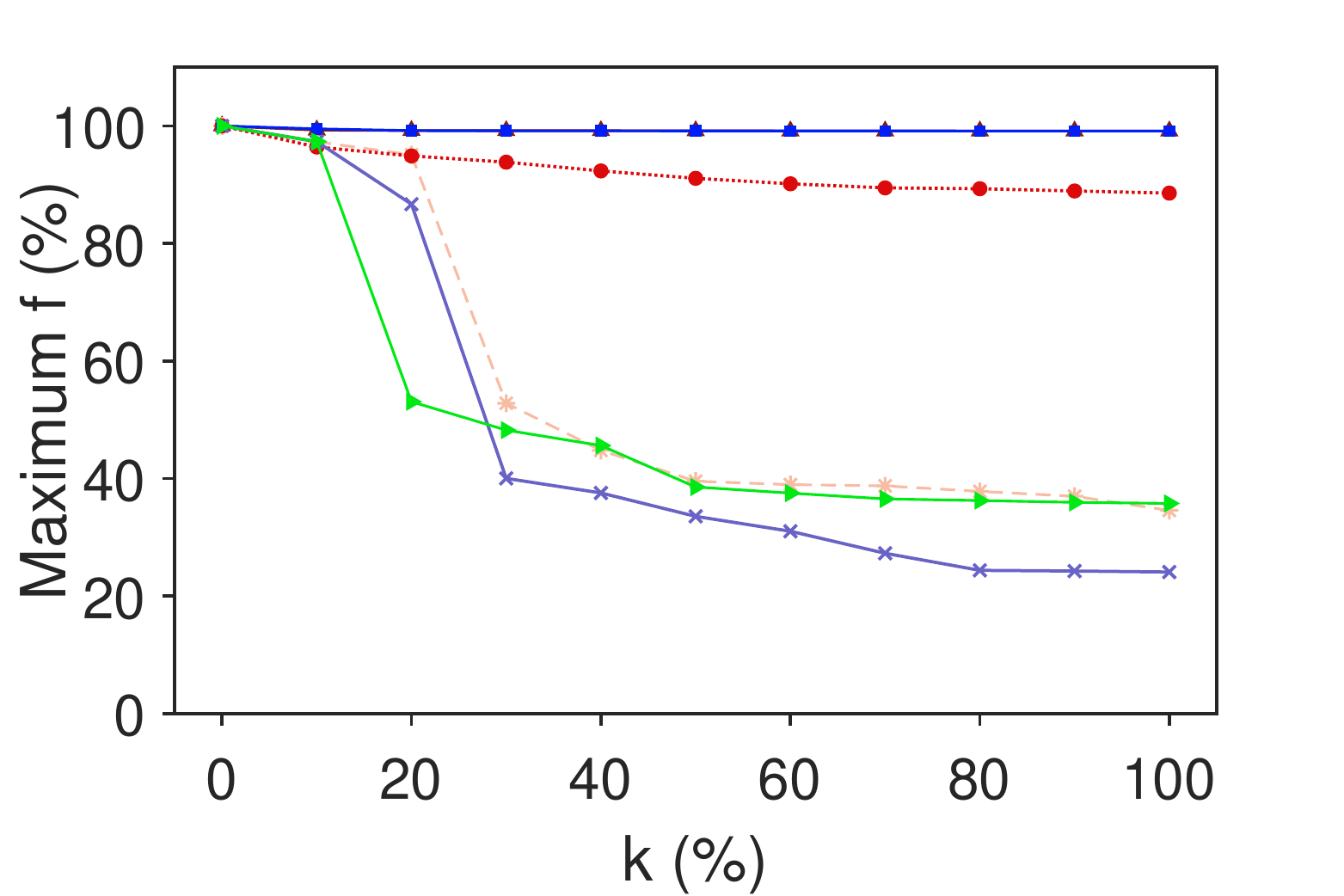}}}
\renewcommand{\thesubfigure}{c}
\subfloat[Amazon Mate]{\scalebox{0.23}{\includegraphics[trim={0 0 1.53cm 1cm},clip]{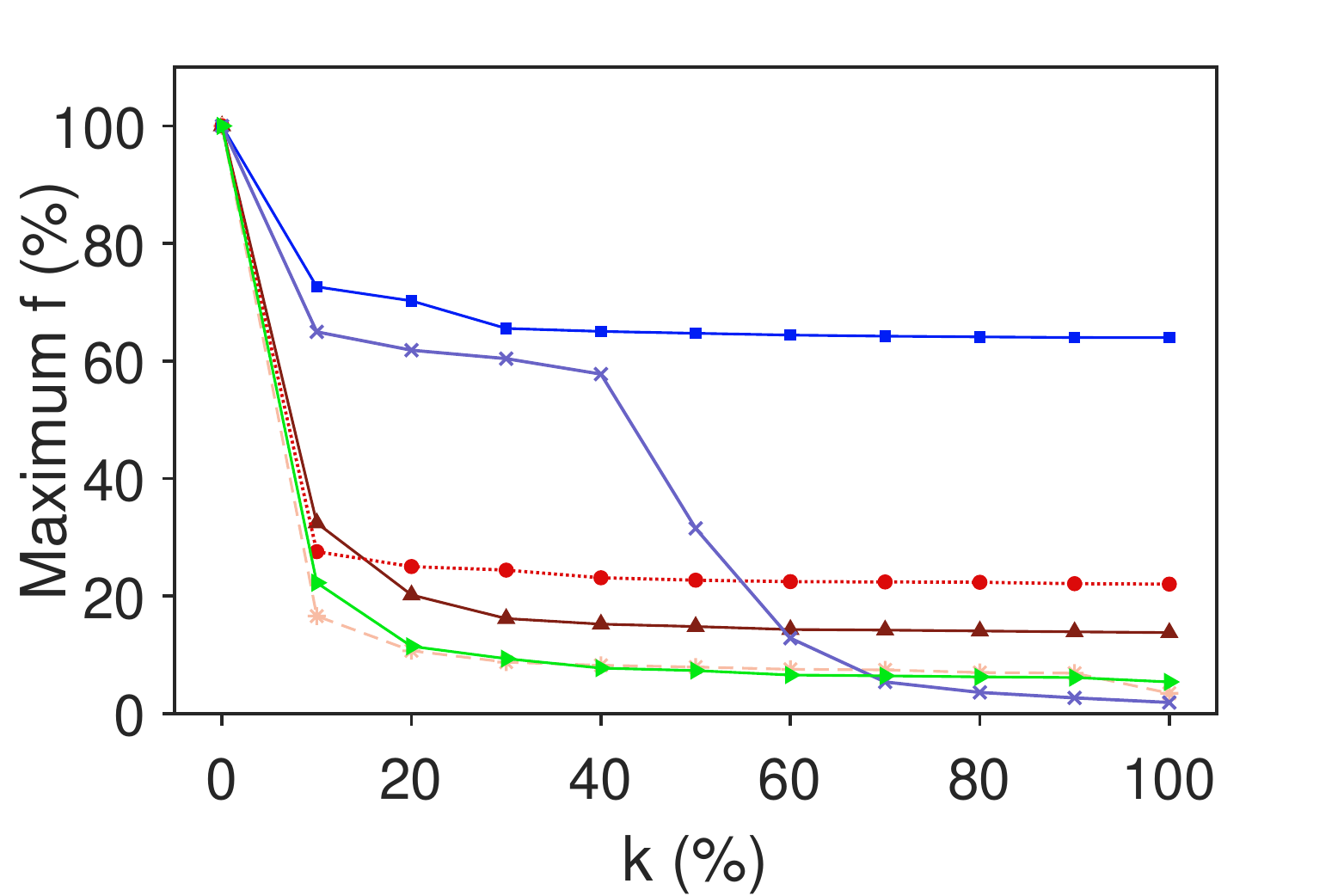}}}
\renewcommand{\thesubfigure}{d}
\subfloat[Amazon MiHi]{\scalebox{0.23}{\includegraphics[trim={0 0 1.57cm 1cm},clip]{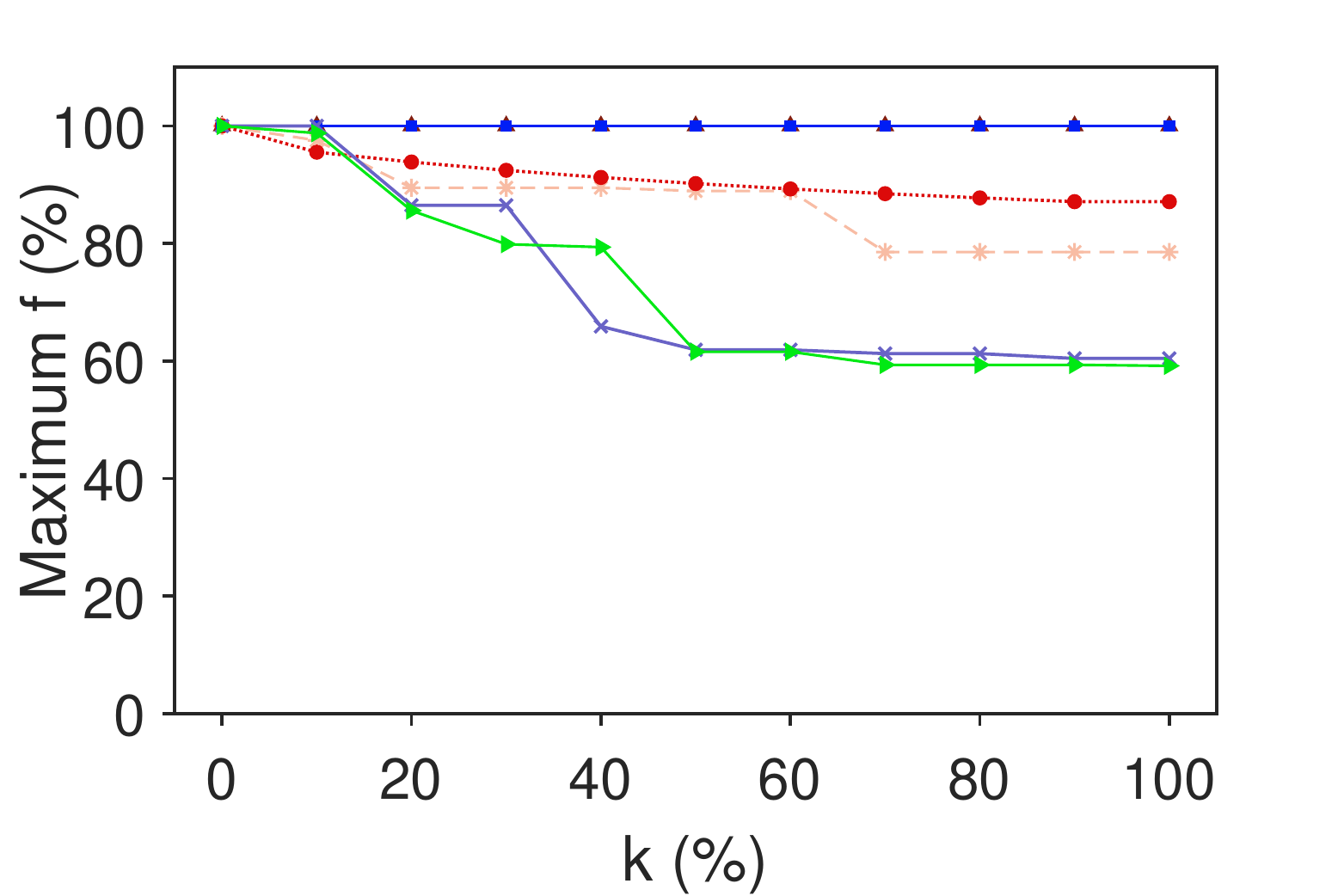}}}
\subfloat[Wiki Sociol.]{\scalebox{0.23}{\includegraphics[trim={0 0 1.57cm 1cm},clip]{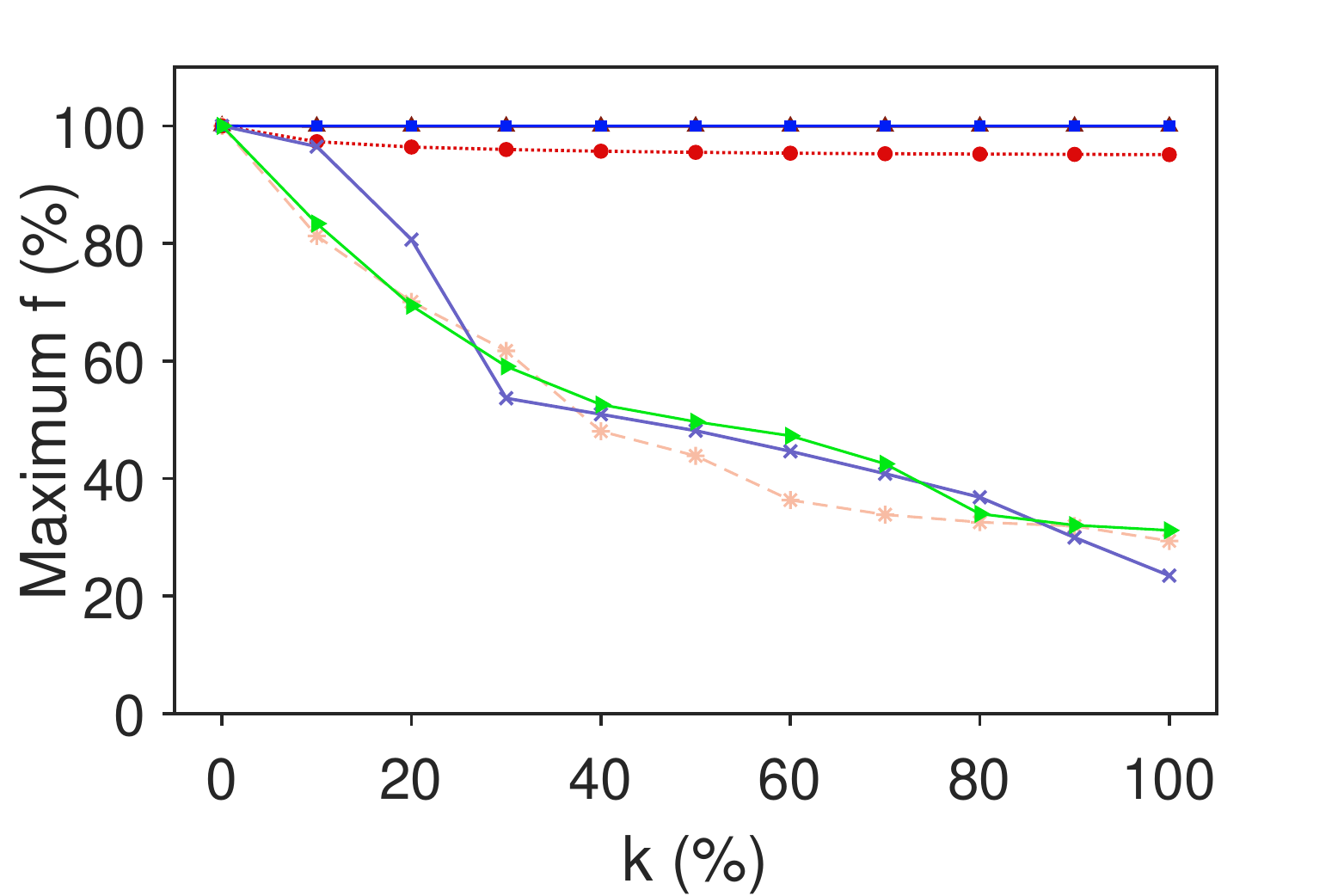}}}
\\
\subfloat{\scalebox{0.5}{\includegraphics[trim={4.2cm 1.3cm 6.5cm 8.4cm},clip]{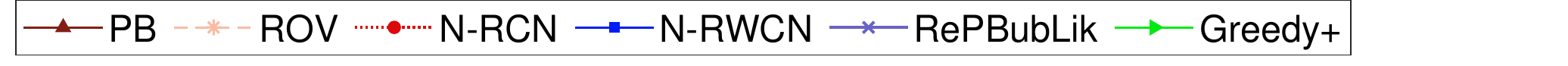}}}
\caption{The y-axis shows the performance of the algorithms in reducing the functions $g$ (first row) and $f$ (second row) in terms of the number of newly added shortcut edges $k$. The x-axis shows the number of new edges $k$ as a fraction of $|R|$. \label{fig:expk}}
\end{figure*}

\section{Experimental Evaluation}
\label{sec:experiments}

Although our work is mostly theoretical, we will compare the performance of our most scalable algorithm  \greedyplus with several baselines on real-life datasets. In particular, we want to measure how \greedyplus reduces the objective functions of \BMAH and \BMMH when we add an incrementally larger set of shortcut edges to the graph, and we want to verify if other existing algorithms are effective for these two tasks.

\pttitle{Baselines} We will compare with the fastest variant of the \texttt{RePBubLik} algorithm, which is the \texttt{RePBubLik+} algorithm \cite{repbub}, as well as with three simplified variants of \texttt{RePBubLik+}. We set the \emph{parochial} nodes \cite{repbub} of the algorithm equal to $R$.
The \texttt{RePBubLik+} algorithm ranks the parochial nodes according to their random-walk centrality (only computed once), and includes a penalty factor that favours the insertion of new shortcut edges to red nodes that have not been shortcut before. The first variant, \emph{PureRandom} \cite{repbub} selects the endpoints of the new edges uniformly at random from $R$ and $B$. The second variant is \emph{Random Top-$N$ Central Nodes ($N$-RCN)} \cite{repbub}, which sorts the top $N$ red nodes in order of descending random-walk centrality and picks $k$ source nodes uniformly from this set of $N$ nodes. The third variant is \emph{Random Top-$N$ Weighted Central Nodes ($N$-RWCN)}, which accounts additionally accounts for the random-walk transition probabilities in the ordering \cite{repbub}. Finally, we compare with the \emph{ROV} algorithm \cite{garimella2017reducing}. ROV outputs $k$ shortcut edges such that their addition aims to minimize the Random Walk Controversy (RWC) score \cite{garimella2018quantifying} of the augmented graph.
The RWC score tries to capture how well separated the two groups are with respect to a certain controversial topic.
It considers candidate edges between high-degree vertices of both groups $R$ and $B$.
Then the top-$k$ edges are picked with respect to the RWC score. 
All parameters of the aforementioned algorithms are set to their standard settings.


\begin{table}[t]
  \caption{Datasets used in the experiments. For each network, we extracted the largest connected component. $|E|_{R \leftrightarrow B}$ denotes the number of inter-group edges, and $|E|_{tot}$ the total number of edges.}
  \label{tab:datasets}
  \small
  \begin{tabular}{lrrcc}
    \toprule
    Data  & $|R|$ & $|B|$ & $|E|_{R \leftrightarrow B}$ & $|E|_{\text{tot}}$    \\
    \midrule
    \emph{Wiki Guns} \cite{menghini2020auditing}  & 134 & 117 & 132 & 550\\
    \emph{Wiki Abort.} \cite{menghini2020auditing}  & 208 & 396 & 232 & 1585\\
    \emph{Amazon Mate} \cite{snapnets}   & 160 & 11 & 16 & 287\\
    \emph{Amazon MiHi} \cite{snapnets}   & 25 & 63 & 56 & 146\\
    \emph{Wiki Sociol.} \cite{menghini2020auditing}  & 648 & 2588 & 430 & 8745\\
  \bottomrule
\end{tabular}
\end{table}

\pttitle{Experimental setup} All experiments are performed on an Intel\,core\,i5 machine at~1.8 GHz with 16\,GB\,RAM. 
Our methods are implemented in Python~3.8 and we made publicly available.\footnote{https://anonymous.4open.science/r/KDD-2023-Source-Code-0E2C/}
We use the same datasets as \texttt{RePBubLik}~\cite{repbub}, see Table~\ref{tab:datasets}.

\pttitle{Implementation of \greedyplus} We set $\epsilon = \lambda = 0.1$ in our experiments. In order to draw a fair comparison with the other baselines, we run \greedyplus for $k$ steps and not more.
We also simplify the way we calculate the average $g$ over all the red nodes in each iteration.
Lemma~\ref{lem:p1} states that we need to run several walks starting from each $r \in R$ and then compute the empirical average, but we only do this for a randomly chosen subset of $|R|/10$ red nodes. This significantly sped up the process, without sacrificing too much in~quality.

\pttitle{Results} Figure~\ref{fig:expk} shows the performance of the proposed algorithms and baselines.
On the $x$-axis it shows the number of new shortcut edges $k$ as a fraction of the total number of red nodes $|R|$. We observe that our algorithm \greedyplus is competitive, performing very similar to \texttt{RePBubLik+}, for minimizing both objective functions $f$ and $g$. Both these algorithms outperform the other baselines. The slightly better performance for \texttt{RePBubLik+} in some cases might be due to specific parameter settings, as well as the fact \texttt{RePBubLik+} is repeated 10 times whereas our algorithm only once. Intuitively \texttt{RePBubLik+} and \greedyplus are in fact expected to perform very similar since they both are practical approximative versions (with different ways of estimating) of the same underlying greedy algorithm: picking the next edge that maximizes the \texttt{RePBubLik+} objective function is theoretically the same edge that \greedy (see Algorithm~\ref{A:GreedyBMAH}) would select. We note that the y-axis shows the \emph{exact} function evaluations of $f$ and $g$, which was done by solving a linear system with $\bigO(n)$ variables (see Sect.~\ref{ss:spg}), which was relatively time consuming and an important reason why we restricted ourselves to smaller datasets.

\section{Conclusion}
\label{sec:conclusion}
In this paper we studied the problem of minimizing average hitting time and maximum hitting time
between two disparate groups in a network by adding new edges between pairs of nodes. 
In contrast to previous methods that modify the objective so that it becomes a submodular function and the optimization becomes straightforward, we minimize hitting time directly. Our approach leads to having a more natural objective 
at the cost of a more challenging optimization problem. For the two problems we define we present several observations and new ideas that lead to novel algorithms with provable approximation guarantees.
For average hitting time we show that the objective is super\-modular
and we apply a known bi\-criteria greedy method;
furthermore, we show how to efficiently approximate the computation of the greedy step 
by sampling bounded-length random walks.
For maximum hitting time, 
we show that it relates to average hitting time, 
and thus, we can reuse the greedy method. 
In addition, we also demonstrate a connection with the 
asymmetric $k$-center problem. 

\begin{acks}
This research is supported by the Helsinki Institute for Information Technology HIIT, Academy of Finland projects AIDA (317085) and MLDB (325117), the ERC Advanced Grant REBOUND (834862), the EC H2020 RIA project SoBigData++ (871042), and the Wallenberg AI, Autonomous Systems and Software Program (WASP) funded by the Knut and Alice Wallenberg Foundation.
\end{acks}

\balance
\bibliographystyle{ACM-Reference-Format}
\bibliography{kdd-hitting}

\end{document}